\documentclass[12pt]{article}
\usepackage{ctex}
\usepackage{amssymb,amsmath,color,bm} 
\usepackage[colorlinks,linkcolor=blue]{hyperref}
\usepackage{geometry}
\usepackage{tabularx}
\usepackage{authblk} 
\usepackage{tikz}
\usetikzlibrary{fit,patterns,arrows.meta,decorations.pathmorphing}
\usepackage{nicematrix}
\usepackage{enumitem}
\usetikzlibrary{angles, quotes} 
\usepackage[square, comma, sort&compress, numbers]{natbib}
\usepackage{float}  

\newcommand{\rank}{{\mathrm{rank}}}
\newcommand{\hull}{{\mathrm{Hull}}}
\newtheorem{theorem}{Theorem}[section]
\newtheorem{definition}{Definition}[section]
\newtheorem{lemma}{Lemma}[section]
\newtheorem{example}{Example}[section]

\newtheorem{corollary}{Corollary}[section]
\newtheorem{remark}{Remark}[section]

\catcode`@=11 \@addtoreset{equation}{section} \catcode`@=12
\NiceMatrixOptions{cell-space-limits
	=1pt}
\begin{document}
\title{\bfseries Non-GRS type Euclidean and Hermitian LCD codes and Their Applications for EAQECCs\thanks{E-mails: liangzhongh0807@163.com; 3120193984@qq.com; qunyingliao@sicnu.edu.cn; cuilingfan@163.com\\ ; zzc@swjtu.edu.cn}} 
\author[1]{\small Zhonghao Liang}\author[1]{\small Dongmei Huang}  
\author[1]{\small Qunying Liao\thanks{Corresponding author}} 
\author[2]{\small Cuiling Fan}
\author[3]{\small Zhengchun Zhou}

\affil[1]{\small College of Mathematical Sciences, Sichuan Normal University, Chengdu 610066, China}

\affil[2]{\small School of Mathematics, Southwest Jiaotong University, Chengdu 611756, China}
\affil[3]{\small School of Information Science and Technology, Southwest Jiaotong University, Chengdu, 611756, China}
	\date{}
	\maketitle
	
{\bf Abstract.}
{\small In recent years, the construction of non-GRS type linear codes has attracted considerable attention due to that they can effectively resist both the Sidelnikov-Shestakov attack and the Wieschebrink attack. Constructing linear complementary dual (LCD) codes and determining the hull of linear codes have long been important topics in coding theory, as they play the crucial role in constructing entanglement-assisted quantum error-correcting codes (EAQECCs), certain communication systems and cryptography. In this paper, by utilizing a class of non-GRS type linear codes, namely, generalized Roth-Lempel (in short, GRL) codes, we firstly construct several classes of Euclidean LCD codes, Hermitian LCD codes, and linear codes with small-dimensional hulls, generalized the main results given by Wu et al. in 2021. We also present an upper bound for the number of a class of Euclidean GRL codes with $1$-dimensional hull, and then for several classes of Hermitian GRL codes, we firstly derive an upper bound for the dimension of the hull, and prove that the bound is attainable. Secondly, as an application, we obtain several families of EAQECCs. Thirdly, we prove that the GRL code is non-GRS for $k>\ell$. Finally, some corresponding examples for LCD MDS codes and LCD NMDS codes are presented.
	}\\
	
	{\bf Keywords.}	{\small Roth-Lempel codes; LCD codes; Hulls; Entanglement-assisted quantum error-correcting codes.}
\section{Introduction}
Let $\mathbb{F}_{q}$ be a finite field with $q$ elements. An $[n, k, d]$ linear code $\mathcal{C}$ over $\mathbb{F}_{q}$ is a linear subspace of $\mathbb{F}_{q}^{n}$ with dimension $k$ and minimum Hamming distance $d$. The Euclidean dual code and the Hermitian dual code of $\mathcal{C}$ are respectively defined by
$$
\mathcal{C}^{\perp_{E}}=\left\{\left(x_{1}, \ldots, x_{n}\right)=\boldsymbol{x}\in\mathbb{F}_q^{n} \mid\langle\boldsymbol{x},\boldsymbol{y}\rangle_{E}=\sum\limits_{i=1}^{n} x_{i} y_{i}=0,  \forall \boldsymbol{y}=\left(y_{1}, \ldots, y_{n}\right) \in \mathcal{C}\right\}$$
and 
$$
\mathcal{C}^{\perp_{H}}=\left\{\left(x_{1}, \ldots, x_{n}\right)=\boldsymbol{x}\in\mathbb{F}_{q^2}^{n} \mid\langle\boldsymbol{x},\boldsymbol{y}\rangle_{H}=\sum\limits_{i=1}^{n} x_{i} y_{i}^{q}=0,  \forall \boldsymbol{y}=\left(y_{1}, \ldots, y_{n}\right) \in \mathcal{C}\right\}.$$

For a linear code $\mathcal{C}$, the hull is defined by $\hull\left(\mathcal{C}\right)=\mathcal{C}\cap\mathcal{C}^{\perp}$, where $\mathcal{C}^{\perp}$ is the dual code of $\mathcal{C}$. The hull plays an important
role for determining the complexity of algorithms to check the permutation equivalence of two linear codes\cite{HullP}, computing the automorphism group of a linear code\cite{HullPG}, calculating the number of shared pairs that required to construct an entanglement-assisted quantum error-correcting code  (EAQECC)\cite{EAQECCHull}, and so on. In particular, these algorithms tend to be highly effective when the dimension of the hull is small. In addition, it is worth mentioning that a special case of the hull of
linear codes is of much interest, i.e., $\hull\left(\mathcal{C}\right)=\left\{\boldsymbol{0}\right\}$, in which $\mathcal{C}$ is called a linear complementary dual (in short, LCD) code. And LCD codes have been widely applied in data
storage, communication systems, and
cryptography\cite{LCDapplication1,LCDapplication2,LCDapplication3}. Thus, determining the value of $\dim\left(\hull\left(\mathcal{C}\right)\right)$, constructing  LCD codes or linear codes with low-dimensional hulls has been interesting \cite{LCD1,LCD2,LCD3,LCD4,Hull1,Hull2,Hull3,Hull4}. 

The well-known Singleton Bound says that $d\leq n-k+1$ for any $[n, k, d]$ code $\mathcal{C}$ over $\mathbb{F}_{q}$, which means that $S\left(\mathcal{C}\right)=n+1-k-d$ is an non-negative integer. If $S\left(\mathcal{C}\right)=0$, then the code $\mathcal{C}$ is maximum distance separable (in short, MDS). If $S\left(\mathcal{C}\right)=1$, then the code $\mathcal{C}$ is almost MDS (in short, AMDS). Especially, if $S\left(\mathcal{C}\right)=S\left(\mathcal{C}^{\perp}\right)=1$, then $\mathcal{C}$ is near MDS (in short, NMDS).

The most well-known class of MDS codes is the generalized Reed-Solomon (in short, GRS) code. If the code $\mathcal{C}$ is not equivalent to any GRS code, then the code $\mathcal{C}$ is called to be non-GRS type. Note that MDS codes constructed from GRS codes are equivalent to GRS codes, and so an MDS code is either a GRS type or a non-GRS type, as the following Figure \ref{fig:mds_classification}.
\begin{figure}[htbp]
	\centering
	\begin{tikzpicture}[
		box/.style={draw, thick, rectangle, minimum height=2cm, minimum width=8cm},
		split line/.style={thick},
		text label/.style={font=\normalsize\bfseries}
		] 
		\node[box] (outer) at (0,0) {}; 
		\draw[split line] (outer.north) -- (outer.south); 
		\node[text label, text=blue] at (-2, 0) {GRS type};
		\node[text label, text=red] at (2, 0) {non-GRS type}; 
		\node[text label, text=black] at (0, -1.5) {MDS codes};
	\end{tikzpicture}
	\caption{Classification of MDS codes} 
	\label{fig:mds_classification}
\end{figure}
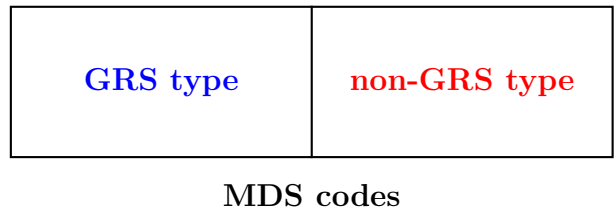\\
Up to now, a lot of results have been obtained on  constructing non-GRS type linear codes\cite{nonGRS1,nonGRS2,nonGRS3,nonGRS4,nonGRS5,nonGRS6,nonGRS7,nonGRS8,Roth1989}.
 
So far, most of LCD codes and EAQECCs have been constructed from a GRS type linear code \cite{LCD1,LCD2,LCD3,LCD4,EAQECC1,EAQECC2,EAQECC3,EAQECC4,EAQECC5}. Thus, constructing LCD codes or EAQECCs is interesting based on non-GRS type linear codes. There have already been some results on constructing Euclidean LCD codes based on non-GRS type linear codes\cite{nonGRSLCD1,nonGRSLCD2,nonGRSLCD3,nonGRSLCD4,nonGRSLCD5}. While, for constructing Hermitian LCD codes and EAQECCs based on non-GRS type linear codes, there are few relevant work  \cite{nonGRSLCD4,nonRSEAQEC}. Especially, in 2021, Wu et al. \cite{nonGRSLCD4} considered the Roth-Lempel (in short, RL) code proposed by Roth R M. and Lempel A.\cite{Roth1989} in 1989, which has the following generator matrix
$$
\begin{pmatrix}\begin{matrix}
		\boldsymbol{G}_{RS}(\boldsymbol{\alpha})
	\end{matrix}&\begin{matrix}
		\boldsymbol{0}_{(k-2)\times 2}\\
		\boldsymbol{T}_{2\times 2}(\delta)
	\end{matrix}
\end{pmatrix}_{k\times (n+2)},
$$
where $\boldsymbol{G}_{RS}(\boldsymbol{\alpha})$ is the generator matrix of the RS code with the evaluation-point sequence $\boldsymbol{\alpha}=\left(\alpha_{1}, \ldots, \alpha_{n}\right) \in \mathbb{F}_{q}^{n}$, and $\boldsymbol{T}_{2\times 2}(\delta)=\begin{pmatrix}
	0&1\\
	1&\delta
\end{pmatrix}$ with $\delta \in\mathbb{F}_{q}$. By taking some special $\boldsymbol{\alpha}$, they proved that there exists Hermitian LCD RL codes. In 2025, the authors\cite{GRL1} introduced the generalized Roth-Lempel (in short, GRL) code, which is a generalization of RL codes, the corresponding linear code over $\mathbb{F}_{q}$ has the generator matrix
$$
\begin{pmatrix}\begin{matrix}
	\boldsymbol{G}_{RS}(\boldsymbol{\alpha})
\end{matrix}&\begin{matrix}
	\boldsymbol{0}_{(k-\ell)\times \ell}\\
	\boldsymbol{A}_{\ell\times\ell}
\end{matrix}
\end{pmatrix}_{k\times (n+\ell)},
$$
where $\boldsymbol{G}_{RS}(\boldsymbol{\alpha})$ is the generator matrix of the RS code with the evaluation-point sequence $\boldsymbol{\alpha}=\left(\alpha_{1}, \ldots, \alpha_{n}\right) \in \mathbb{F}_{q}^{n}$, and $\boldsymbol{A}_{\ell\times\ell}\in\mathrm{GL}_{l}\left(\mathbb{F}_{q}\right)$. So far, for several special matrices $\boldsymbol{A}_{\ell\times\ell}$, the corresponding MDS property\cite{GRL1,nonGRS8}, AMDS property\cite{GRL1,nonGRS8}, NMDS property\cite{GRL2,RLNMDS1,RLNMDS2}, self-dual property\cite{GRL1,nonGRS8}, extendable property\cite{GRL3,nonGRS8}, the existence of LCD codes\cite{nonGRSLCD4} and decoding algorithms of punctured codes \cite{RLdeconding} have been investigated, respectively.

In this paper, different from the work of Wu et al. (2021), which only considered a very special class of $2\times 2$ matrices $\boldsymbol{T}_{2\times 2}(\delta)$ and some special $\boldsymbol{\alpha}$, and they only proved that there exists LCD RL codes, we extend their research to the most general $\ell\times\ell$ matrix $\boldsymbol{A}_{\ell\times\ell}$ and more flexible $\boldsymbol{\alpha}$, and give the specific construction of LCD codes. We construct several classes of Euclidean LCD codes, Hermitian LCD codes, small-dimensional hull linear codes and EAQECCs. And for some Hermitian GRL codes, we also obtain an upper bound for the dimension of the hull and prove that the bound is attainable. 

This paper is organized as follows. In Section \ref{sec2}, we give the definition of the GRL code and some necessary lemmas. In Sections \ref{sec3}-\ref{sec4}, we give some constructions for Euclidean LCD GRL codes, Hermitian LCD GRL codes, small-dimensional hull linear codes and EAQECCs. In Section \ref{sec5}, we prove that the GRL code is non-GRS.
 In Section \ref{sec6}, we conclude the whole paper. In Appendix, some examples are given.
\section{Preliminaries}\label{sec2}
Throughout this paper, for the convenience, we fix some notations as the following.
\begin{itemize}
	\item $q$ is a power of an odd prime.
	\item $\mathbb{F}_{q}$ or $\mathbb{F}_{q^2}$ is the finite field with $q$ or $q^{2}$ elements, respectively. 
	\item For any non-empty set $\left\{x_{1},\ldots,x_{n}\right\}$, $\left\{x_{1},\ldots,x_{n}\right\}_{\bmod\ k}\triangleq\left\{x_{1}(\bmod\ k),\ldots,x_{n}(\bmod\ k)\right\}$.
	\item For any set $A$, $\# A$ denotes the number of elements in $A$.
	\item $d(x)$ denotes the number of positive divisors of the positive integer $x$.
	\item For the prime number $p$ and the positive integer $x$, $v_{p}(x)$ denotes the largest non-negative integer $k$ such that $p^{k}\mid x$ and $p^{k+1}\nmid x$.
	\item $\gcd\left(a,b\right)$ denotes the greatest common divisor for two positive integers $a$ and $b$.
	\item $\dim\left(\mathrm{Hull}_{E}(\mathcal{C})\right)$ or $\dim\left(\mathrm{Hull}_{H}(\mathcal{C})\right)$ denotes the dimension of the Euclidean hull or the Hermitian hull for the linear code $\mathcal{C}$, respectively.
	\item For any matrix $\boldsymbol{G}$, $\overline{\boldsymbol{G}}$ denotes the conjugate matrix of $\boldsymbol{G}$. 
\end{itemize} 

In this section, we recall the definition of the generalized Roth-Lempel code and some necessary lemmas.
\begin{definition}\label{GRLdefinition}{\rm(\cite{GRL1}, Definition 1)}
Let $3\leq \ell+1\leq k+1\leq n\leq q$,  $\boldsymbol{\alpha}=\left(\alpha_{1}, \ldots, \alpha_{n}\right) \in \mathbb{F}_{q}^{n}$ with $\alpha_{i} \neq \alpha_{j}(i \neq j)$ and $\boldsymbol{v}=$ $\left(v_{1}, \ldots, v_{n}\right) \in\left(\mathbb{F}_{q}^{*}\right)^{n}$. The generalized Roth-Lempel (in short, GRL) code $\mathrm{GRL}_{k}(\boldsymbol{\alpha}, \boldsymbol{v},\boldsymbol{A}_{\ell\times \ell})$ is defined as
	$$
	\mathrm{GRL}_{k}(\boldsymbol{\alpha}, \boldsymbol{v},\boldsymbol{A}_{\ell\times \ell})\triangleq\left\{\left(v_{1} f\left(\alpha_{1}\right), \ldots, v_{n} f\left(\alpha_{n}\right),\boldsymbol{\beta}\right) | f(x) \in \mathbb{F}_{q}^{k}[x]\right\},
	$$
	where $\boldsymbol{A}_{\ell\times \ell}=(a_{ij})_{\ell\times \ell}\in\mathrm{GL}_{\ell}\left(\mathbb{F}_{q}\right)$ and 
	$$\begin{aligned}  \boldsymbol{\beta}=&\left(f_{k-\ell},\ldots,f_{k-1}\right)\boldsymbol{A}_{\ell\times \ell}\\
		=&\left(a_{11}f_{k-\ell}+a_{21}f_{k-(\ell-1)}+\cdots+a_{\ell1}f_{k-1},\ldots,a_{1\ell}f_{k-\ell}+a_{2l}f_{k-(\ell-1)}+\cdots+a_{\ell\ell}f_{k-1}\right).
	\end{aligned}$$
\end{definition}
\begin{lemma}\label{monomiallyequivalent}{\rm(\cite{nonGRS8}, Definition 3)}
Let $\mathcal{C}_{1}$ and $\mathcal{C}_{2}$ be two linear codes  of the same code length over $\mathbb{F}_{q}$, and let $\boldsymbol{M}$ be a generator matrix of $\mathcal{C}_{1}$. Then $\mathcal{C}_{1}$ and $\mathcal{C}_{2}$ are monomially equivalent if and only if there exists a monomial matrix $\boldsymbol{D}$ such that $\boldsymbol{MD}$ is a generator matrix of $\mathcal{C}_{2}$.
\end{lemma}
\begin{remark}\label{GRLmonomial}
By Definition \ref{GRLdefinition}, the code $\mathrm{GRL}_{k}(\boldsymbol{\alpha}, \boldsymbol{v},\boldsymbol{A}_{\ell\times \ell})$ and the code $\mathrm{GRL}_{k}(\boldsymbol{\alpha}, \boldsymbol{1},\boldsymbol{A}_{l\times l})$ have the generator matrix
\begin{equation}\label{GRLvgeneratormatrix}
\boldsymbol{G}_{\boldsymbol{v},n}=\begin{pmatrix}
	v_{1}&\cdots&v_{n}&0&\cdots&0\\
	v_{1}\alpha_{1}&\cdots&v_{n}\alpha_{n}&0&\cdots&0\\
	\vdots&\ddots&\vdots&\vdots&\ddots&\vdots\\
	v_{1}\alpha_{1}^{k-(\ell+1)}&\cdots&v_{n}\alpha_{n}^{k-(\ell+1)}&0&\cdots&0\\
	v_{1}\alpha_{1}^{k-\ell}&\cdots&v_{n}\alpha_{n}^{k-\ell}&a_{11}&\cdots&a_{1\ell}\\
	\vdots&\ddots&\vdots&\vdots&\ddots&\vdots\\
	v_{1}\alpha_{1}^{k-1}&\cdots&v_{n}\alpha_{n}^{k-1}&a_{\ell1}&\cdots&a_{\ell\ell}\\
\end{pmatrix}
\end{equation}
and 
\begin{equation}\label{GRL1generatormatrix}
\boldsymbol{G}_{\boldsymbol{1},n}=\begin{pmatrix}
	1&\cdots&1&0&\cdots&0\\
	\alpha_{1}&\cdots&\alpha_{n}&0&\cdots&0\\
	\vdots&\ddots&\vdots&\vdots&\ddots&\vdots\\
	\alpha_{1}^{k-(\ell+1)}&\cdots&\alpha_{n}^{k-(\ell+1)}&0&\cdots&0\\
	\alpha_{1}^{k-\ell}&\cdots&\alpha_{n}^{k-\ell}&a_{11}&\cdots&a_{1\ell}\\
	\vdots&\ddots&\vdots&\vdots&\ddots&\vdots\\
	\alpha_{1}^{k-1}&\cdots&\alpha_{n}^{k-1}&a_{\ell1}&\cdots&a_{\ell\ell}\\
\end{pmatrix},
\end{equation}
respectively. And so by Lemma \ref{monomiallyequivalent}, the code $\mathrm{GRL}_{k}(\boldsymbol{\alpha}, \boldsymbol{v},\boldsymbol{A}_{\ell\times \ell})$ and the code $\mathrm{GRL}_{k}(\boldsymbol{\alpha}, \boldsymbol{1},\boldsymbol{A}_{l\times l})$ are monomially equivalent.
\end{remark}

The following Lemma $\ref{LCDequivalent}$ gives a characterization for a linear code to be an Euclidean LCD code or a Hermitian LCD code.
\begin{lemma}\label{LCDequivalent} {\rm(\cite{LCD1}, Proposition 2)} If $\boldsymbol{G}$ is a generator matrix for the $[n,k]$ linear code $\mathcal{C}$ over $\mathbb{F}_{q}$(resp.$\mathbb{F}_{q^2}$), then $\mathcal{C}$ is an Euclidean (resp. Hermitian) LCD code if and only if the $k\times k$ matrix $\boldsymbol{G}\boldsymbol{G}^{T}$ (resp. $\overline{\boldsymbol{G}}$) is nonsingular over $\mathbb{F}_{q}$(resp.$\mathbb{F}_{q^2}$).
\end{lemma} 

The following Lemma $\ref{Fq*sum}$ is very important for our constructions.
\begin{lemma}\label{Fq*sum}\rm(\cite{LCD1})
Let $\mathbb{F}_{q}^{*}=\langle\gamma\rangle$(resp. $\mathbb{F}_{q^2}^{*}=\langle\gamma\rangle$), $s$ be a positive integer with $s\mid q-1$(resp. $s\mid q^2-1$), and $\alpha_{i}=\gamma^{\frac{q-1}{s}i}$ (resp. $\alpha_{i}=\gamma^{\frac{q^2-1}{s}i}$) for $1\leq i\leq k$, then for any integer $t$ and $\beta\in\mathbb{F}_{q}^{*}$(resp. $\beta\in\mathbb{F}_{q^2}^{*}$), we have
$$\sum\limits_{i=1}^{s}\left(\beta\alpha_{i}\right)^{t}=\begin{cases}
	\beta^{t}s,&\text{if}\ s\mid t;\\
	0,&\text{otherwise}.
\end{cases}$$ 
 \end{lemma}
 
The following Lemma $\ref{nonGRS}$ provides a method to determine whether a linear code is non-RS type.
\begin{lemma}\label{nonGRS}
	{\rm(\cite{Roth1985}, Theorem 1)} Let $\boldsymbol{\alpha}=\left(\alpha_{1},\ldots\alpha_{n}\right) \in \mathbb{F}_{q}^{n}$ with $\alpha_{i} \neq \alpha_{j}(i \neq j)$. Suppose that $\boldsymbol{B}$ is a $k \times(n-k)$ matrix and $\boldsymbol{G}=\left(\boldsymbol{E}_{k} |\boldsymbol{B}\right)$ is a $k \times n$ matrix over $\mathbb{F}_{q}$, where $\boldsymbol{E}_{k}$ is the $k \times k$ identity matrix. Then $\boldsymbol{G}$ generates the $\mathrm{RS}$ code $\operatorname{RS}_{k}(\boldsymbol{\alpha})$ if and only if for $1 \leq i \leq k$ and $1 \leq j \leq n-k$, the $(i, j)$-th entry of $\boldsymbol{B}$ is given by $$\frac{\eta_{k+j} \eta_{i}^{-1}}{\alpha_{k+j}-\alpha_{i}},$$ where
	$\eta_{i}=\prod\limits_{s=1, s \neq i}^{k}\left(\alpha_{i}-\alpha_{s}\right)$ and $ \eta_{k+j}=\prod\limits_{s=1}^{k}\left(\alpha_{k+j}-\alpha_{s}\right).$
\end{lemma} 
 
 The following Lemma \ref{nonzerosFqsolution} provides an explicit formula for the number of $k$-tuples $\left(x_1,\ldots,x_{k}\right)\in\mathbb{F}_{q}^{k}$ such that $x_{1}^{2}+\cdots+x_{k}^{2}=c\in\mathbb{F}_{q}$ holds.
\begin{lemma}\label{nonzerosFqsolution}\rm(\cite{FengRQbook2025})
For any element $c\in\mathbb{F}_{q}$, let $v(c)=\begin{cases}
q-1,&c=0;\\
-1,&c\in\mathbb{F}_{q}^{*}
\end{cases}$, and $N_{f}\left(k,c,q\right)$ denote the number of $k$-tuples $\left(x_1,\ldots,x_{k}\right)\in\mathbb{F}_{q}^{k}\backslash\left\{\boldsymbol{0}\right\}$ such that $x_{1}^{2}+\cdots+x_{k}^{2}=c$ holds, then
$$
N_{f}\left(k,c,q\right)=\begin{cases}
q^{n-1}+v(b)q^{\frac{k}{2}-1}\eta\left((-1)^{\frac{k}{2}}\right),&2\mid k;\\
q^{n-1}+v(b)q^{\frac{k-1}{2}}\eta\left((-1)^{\frac{k-1}{2}}c\right),&2\nmid k.\\
\end{cases}
$$
\end{lemma}

In particular, for $\left(x_1,\ldots,x_{k}\right)\in\left(\mathbb{F}_{q}^{*}\right)^{k}$, Feng et al. \cite{FengRQ2025} gave the following Lemma \ref{nonzerosFq*solution}.
\begin{lemma}\label{nonzerosFq*solution}\rm (\cite{FengRQ2025}, Theorems 2.6-2.7)
For any element $c\in\mathbb{F}_{q}$, let $N_{f}^{*}\left(k,c,q\right)$ denote the number of $k$-tuples $\left(x_1,\ldots,x_{k}\right)\in\left(\mathbb{F}_{q}^{*}\right)^{k}$ such that $x_{1}^{2}+\cdots+x_{k}^{2}=c$ holds, then the following statements are true,

$(1)$ if $q \equiv 1 \pmod{4}$, then
$$
N_{f}^{*}(k, c, q) =
\begin{cases}
	\dfrac{2(q-1)^k + (q-1)\left((\sqrt{q}-1)^k + (-1-\sqrt{q})^k\right)}{2q}, & \text{if } c = 0; \\
	\dfrac{2(q-1)^k + (\sqrt{q}-1)^{k+1} + (-1-\sqrt{q})^{k+1}}{2q}, & \text{if } c = a^2 \text{ for some } a \in \mathbb{F}_q^*; \\
	\dfrac{2(q-1)^k + (1-q)\left((\sqrt{q}-1)^{k-1} + (-1-\sqrt{q})^{k-1}\right)}{2q}, & \text{otherwise};
\end{cases}
$$

$(2)$ if $q \equiv 3 \pmod{4}$, then
$$
N_{f}^{*}(k, c, q) =
\begin{cases}
	\dfrac{2(q-1)^k + (q-1)\left((\sqrt{-q}-1)^k + (-1-\sqrt{-q})^k\right)}{2q}, & \text{if } c = 0; \\
	\dfrac{2(q-1)^k + (q+1)\left((\sqrt{-q}-1)^{k-1} + (-1-\sqrt{-q})^{k-1}\right)}{2q}, & \text{if } c = a^2 \text{ for some } a \in \mathbb{F}_q^*; \\
	\dfrac{2(q-1)^k + (\sqrt{-q}-1)^{k+1} + (-1-\sqrt{-q})^{k+1}}{2q}, & \text{otherwise}.
\end{cases}
$$
\end{lemma}

For an $\left[n,k,d\right]$ linear code, the following Lemmas \ref{Ehull}-\ref{EAQECCsHhull} provide a method for calculating the dimension of the Euclidean hull or the Hermitian hull, and constructing EAQECCs.
\begin{lemma}\label{Ehull}\rm (\cite{EAQECCHull}, Proposition 3.1)
	Let $\mathcal{C}$ be an $\left[n,k,d\right]_{q}$ linear code with the generator matrix $\boldsymbol{G}$ and the parity
	check matrix $\boldsymbol{H}$. Then, we have
	$$\rank\left(\boldsymbol{H}\boldsymbol{H}^{\perp_{E}}\right)=n-k-\dim\left(\mathrm{Hull}_{E}\left(\mathcal{C}\right)\right)$$
	and 
	$$\rank\left(\boldsymbol{G}\boldsymbol{G}^{\perp_{E}}\right)=k-\dim\left(\mathrm{Hull}_{E}\left(\mathcal{C}^{\perp_{E}}\right)\right).$$
\end{lemma}
\begin{lemma}\label{EAQECCsEhull}\rm (\cite{EAQECCHull}, Corollary 3.1)
	Let $\mathcal{C}$ and $\mathcal{C}^{\perp_{E}}$ be the classical linear code and its Euclidean dual code
	with the parameters $\left[n,k,d\right]_{q}$ and $\left[n,k,d^{\perp_{E}}\right]_{q}$, respectively. Then there exist two EAQECCs with the parameters $$\left[\left[n,k-\dim\left(\mathrm{Hull}_{E}\left(\mathcal{C}\right)\right),d,n-k-\dim\left(\mathrm{Hull}_{E}\left(\mathcal{C}\right)\right)\right]\right]_{q}$$ and $$\left[\left[n,n-k-\dim\left(\mathrm{Hull}_{E}\left(\mathcal{C}\right)\right),d^{\perp_{E}},k-\dim\left(\mathrm{Hull}_{E}\left(\mathcal{C}\right)\right)\right]\right]_{q},$$
	respectively. Moreover, if $\mathcal{C}$ is MDS, then the above two EAQECCs are also MDS.
\end{lemma}
\begin{lemma}\label{Hhull}\rm (\cite{EAQECCHull}, Proposition 3.2)
	Let $\mathcal{C}$ be the $\left[n,k,d\right]_{q^2}$ linear code with the generator matrix $\boldsymbol{G}$ and the parity 	check matrix $\boldsymbol{H}$. Then, we have
	$$\rank\left(\boldsymbol{H}\boldsymbol{H}^{\perp_{H}}\right)=n-k-\dim\left(\mathrm{Hull}_{H}\left(\mathcal{C}\right)\right)$$
	and 
	$$\rank\left(\boldsymbol{G}\boldsymbol{G}^{\perp_{H}}\right)=k-\dim\left(\mathrm{Hull}_{H}\left(\mathcal{C}^{\perp_{H}}\right)\right).$$
\end{lemma}
\begin{lemma}\label{EAQECCsHhull}\rm (\cite{EAQECCHull}, Corollary 3.2)
	Let $\mathcal{C}$ and $\mathcal{C}^{\perp_{H}}$ be a classical linear code and its Hermitian dual
	with the parameters $\left[n,k,d\right]_{q^2}$ and $\left[n,k,d^{\perp_{H}}\right]_{q^2}$, respectively. Then there exists two EAQECCs with the parameters $$\left[\left[n,k-\dim\left(\mathrm{Hull}_{H}\left(\mathcal{C}\right)\right),d,n-k-\dim\left(\mathrm{Hull}_{H}\left(\mathcal{C}\right)\right)\right]\right]_{q}$$ and $$\left[\left[n,n-k-\dim\left(\mathrm{Hull}_{H}\left(\mathcal{C}\right)\right),d^{\perp_{H}},k-\dim\left(\mathrm{Hull}_{H}\left(\mathcal{C}\right)\right)\right]\right]_{q},$$
	respectively. Moreover, if $\mathcal{C}$ is MDS, then the above two EAQECCs are also MDS.
\end{lemma}
\section{Euclidean LCD GRL codes and their applications}\label{sec3}
 Throughout this section, we fix $k\mid q-1$, $\mathbb{F}_{q}^{*}=\mathbb{F}_{q}\backslash\left\{0\right\}=\langle\gamma\rangle$ and $\alpha_{i}=\gamma^{\frac{q-1}{k}i}(1\leq i\leq k)$. In this section, by taking a special vector $\boldsymbol{\alpha}=\left(\alpha_{1},\ldots,\alpha_{n}\right)$, we construct four classes of Euclidean LCD GRL codes with the parameters $\left[n+\ell,k\right]$, get some GRL codes with small-dimensional hull, and then obtain several families of EAQECCs.

\subsection{The first class of LCD GRL codes with the parameters $\left[k+\ell,k\right]$}
In this subsection, by taking $\boldsymbol{\alpha}=\left(\gamma^{\delta}\alpha_{1},\ldots,\gamma^{\delta}\alpha_{k}\right)$, we construct two classes of Euclidean LCD GRL codes, get a class of GRL codes with $1$-dimensional hull, and then obtain an upper bound for the number of GRL codes with $1$-dimensional hull.
\begin{theorem}\label{EuclideanLCDRL1}
Let $\boldsymbol{\alpha}=\left(\gamma^{\delta}\alpha_{1},\ldots,\gamma^{\delta}\alpha_{k}\right)$ with $1\leq \delta\leq q-1$. Then the following two statements are true,
	
	$(1)$ if $\ell<\frac{k}{2} $, then the code $\mathrm{GRL}_{k}(\boldsymbol{\alpha}, \boldsymbol{v},\boldsymbol{A}_{\ell\times \ell})$ is Euclidean LCD;
	
	$(2)$ if $\ell=\frac{k}{2}$ and $\gamma^{\delta k}k+\sum\limits_{i=1}^{\ell}a_{1i}^{2}\in\mathbb{F}_{q}^{*}$, then the code $\mathrm{GRL}_{k}(\boldsymbol{\alpha}, \boldsymbol{v},\boldsymbol{A}_{\ell\times \ell})$ is Euclidean LCD.
\end{theorem}
{\bf Proof}. By Remark \ref{GRLmonomial}, we only focus on  $\mathrm{GRL}_{k}(\boldsymbol{\alpha}, \boldsymbol{1},\boldsymbol{A}_{l\times l})$ with the generator matrix $\boldsymbol{G}_{\boldsymbol{1},k}$ given by $(\ref{GRL1generatormatrix})$. Furthermore, by Lemma \ref{LCDequivalent}, we only need to prove that the $k\times k$ matrix $\boldsymbol{G}_{\boldsymbol{1},k}\boldsymbol{G}_{\boldsymbol{1},k}^{T}$ is nonsingular over $\mathbb{F}_{q}$, i.e.,   $\rank\left(\boldsymbol{G}_{\boldsymbol{1},k}\boldsymbol{G}_{\boldsymbol{1},k}^{T}\right)=k$. In fact, note that 
$$\scriptsize\begin{aligned}	
&\boldsymbol{G}_{\boldsymbol{1},k}\boldsymbol{G}_{\boldsymbol{1},k}^{T}\\
=&\begin{pmatrix}
		k&\sum\limits_{i=1}^{k}\left(\gamma^{\delta}\alpha_{i}\right)&\cdots&\sum\limits_{i=1}^{k}\left(\gamma^{\delta}\alpha_{i}\right)^{k-\ell-1}&\sum\limits_{i=1}^{k}\left(\gamma^{\delta}\alpha_{i}\right)^{k-\ell}&\cdots&\sum\limits_{i=1}^{k}\left(\gamma^{\delta}\alpha_{i}\right)^{k-1}\\
		\sum\limits_{i=1}^{k}\left(\gamma^{\delta}\alpha_{i}\right)&\sum\limits_{i=1}^{k}\left(\gamma^{\delta}\alpha_{i}\right)^{2}&\cdots&\sum\limits_{i=1}^{k}\left(\gamma^{\delta}\alpha_{i}\right)^{k-\ell}&\sum\limits_{i=1}^{k}\left(\gamma^{\delta}\alpha_{i}\right)^{k-\ell+1}&\cdots&\sum\limits_{i=1}^{k}\left(\gamma^{\delta}\alpha_{i}\right)^{k}\\
		\vdots&\vdots&\ddots&\vdots&\vdots&\ddots&\vdots\\
		\sum\limits_{i=1}^{k}\left(\gamma^{\delta}\alpha_{i}\right)^{k-(\ell+1)}&\sum\limits_{i=1}^{k}\left(\gamma^{\delta}\alpha_{i}\right)^{k-\ell}&\cdots&\sum\limits_{i=1}^{k}\left(\gamma^{\delta}\alpha_{i}\right)^{2k-2\ell-2}&\sum\limits_{i=1}^{k}\left(\gamma^{\delta}\alpha_{i}\right)^{2k-2\ell-1}&\cdots&\sum\limits_{i=1}^{k}\left(\gamma^{\delta}\alpha_{i}\right)^{2k-\ell-2}\\
		\sum\limits_{i=1}^{k}\left(\gamma^{\delta}\alpha_{i}\right)^{k-\ell}&\sum\limits_{i=1}^{k}\left(\gamma^{\delta}\alpha_{i}\right)^{k-\ell+1}&\cdots&\sum\limits_{i=1}^{k}\left(\gamma^{\delta}\alpha_{i}\right)^{2k-2\ell-1}&\sum\limits_{i=1}^{k}\left(\gamma^{\delta}\alpha_{i}\right)^{2k-2\ell}+\sum\limits_{i=1}^{\ell}a_{1i}^{2}&\cdots&\sum\limits_{i=1}^{k}\left(\gamma^{\delta}\alpha_{i}\right)^{2k-\ell-1}+\sum\limits_{i=1}^{\ell}a_{1i}a_{\ell i}\\
		\vdots&\vdots&\ddots&\vdots&\vdots&\ddots&\vdots\\
		\sum\limits_{i=1}^{k}\left(\gamma^{\delta}\alpha_{i}\right)^{k-1}&\sum\limits_{i=1}^{k}\left(\gamma^{\delta}\alpha_{i}\right)^{k}&\cdots&\sum\limits_{i=1}^{k}\left(\gamma^{\delta}\alpha_{i}\right)^{2k-\ell-2}&\sum\limits_{i=1}^{k}\left(\gamma^{\delta}\alpha_{i}\right)^{2k-\ell-1}+\sum\limits_{i=1}^{\ell}a_{\ell i}a_{1i }&\cdots&\sum\limits_{i=1}^{k}\left(\gamma^{\delta}\alpha_{i}\right)^{2k-2}+\sum\limits_{i=1}^{\ell}a_{\ell i}^{2}\\
	\end{pmatrix}.
\end{aligned}
$$ 

{\bf For (1).} By $\ell<\frac{k}{2}$, we have $2k-2\ell>k$, furthermore, by Lemma \ref{Fq*sum}, we can get 
$$\footnotesize\boldsymbol{G}_{\boldsymbol{1},k}\boldsymbol{G}_{\boldsymbol{1},k}^{T}=\begin{pmatrix}
		k&0&\cdots&0&0&\cdots&0&0&\cdots&0&0&\cdots&0\\
		0&0&\cdots&0&0&\cdots&0&0&\cdots&0&0&\cdots&\gamma^{\delta k}k\\
		\vdots&\vdots& &\vdots&\vdots& &\vdots&\vdots& &\vdots&\vdots& &\vdots\\ 
		0&0&\cdots&0&0&\cdots&0&0&\cdots&0&\gamma^{\delta k}k&\cdots&0\\ 
		0&0&\cdots&0&0&\cdots&0&0&\cdots&\gamma^{\delta k}k&0&\cdots&0\\
		\vdots&\vdots& &\vdots&\vdots& &\vdots&\vdots& &\vdots&\vdots& &\vdots\\  
		0&0&\cdots&0&0&\cdots&0&\gamma^{\delta k}k&\cdots&0&0&\cdots&0\\ 
		0&0&\cdots&0&0&\cdots&\gamma^{\delta k}k&0&\cdots&0&0&\cdots&0\\
		\vdots&\vdots& &\vdots&\vdots& &\vdots&\vdots& &\vdots&\vdots& &\vdots\\
		0&0&\cdots&0&\gamma^{\delta k}k&\cdots&0&0&\cdots&0&0&\cdots&0\\
		0&0&\cdots&\gamma^{\delta k}k&0&\cdots&0&0&\cdots&0&\sum\limits_{i=1}^{\ell}a_{1i}^{2}&\cdots&\sum\limits_{i=1}^{\ell}a_{1i}a_{\ell i}\\
		\vdots&\vdots& &\vdots&\vdots& &\vdots&\vdots& &\vdots&\vdots& &\vdots\\
		0&\gamma^{\delta k}k&\cdots&0&0&\cdots&0&0&\cdots&0&\sum\limits_{i=1}^{\ell}a_{\ell i}a_{1i }&\cdots&\sum\limits_{i=1}^{\ell}a_{\ell i}^{2}\\
	\end{pmatrix}.$$
Now from $k\mid q-1$ and $\mathbb{F}_{q}^{*}=\langle\gamma\rangle$, we have $\gamma^{\delta k}k\in\mathbb{F}_{q}^{*}$, and so $\rank\left(\boldsymbol{G}_{\boldsymbol{1},k}\boldsymbol{G}_{\boldsymbol{1},k}^{T}\right)=k$. Furthermore, by Lemma \ref{LCDequivalent}, the code $\mathrm{GRL}_{k}(\boldsymbol{\alpha}, \boldsymbol{v},\boldsymbol{A}_{l\times l})$ is Euclidean LCD.

{\bf For (2).} By $\ell=\frac{k}{2}$, we can get $2k-2\ell=k$, furthermore, by Lemma \ref{Fq*sum}, we have 
$$\boldsymbol{G}_{\boldsymbol{1},k}\boldsymbol{G}_{\boldsymbol{1},k}^{T}\\
=\begin{pmatrix}
	k&0&\cdots&0&0&0&\cdots&0\\
	0&0&\cdots&0&0&0&\cdots&\gamma^{\delta k}k\\
	\vdots&\vdots& &\vdots&\vdots&\vdots& &\vdots\\ 
	0&0&\cdots&0&0&\gamma^{\delta k}k&\cdots&0\\
	0&0&\cdots&0&\gamma^{\delta k}k+\sum\limits_{i=1}^{\ell}a_{1i}^{2}&\sum\limits_{i=1}^{\ell}a_{1i}a_{2i}&\cdots&\sum\limits_{i=1}^{\ell}a_{1i}a_{\ell i}\\
	0&0&\cdots&\gamma^{\delta k}k&\sum\limits_{i=1}^{\ell}a_{2i}a_{1i}&\sum\limits_{i=1}^{\ell}a_{2i}^{2}&\cdots&\sum\limits_{i=1}^{\ell}a_{2i}a_{\ell i}\\
	\vdots&\vdots& &\vdots&\vdots&\vdots& &\vdots\\
	0&\gamma^{\delta k}k&\cdots&0&\sum\limits_{i=1}^{\ell}a_{\ell i}a_{1i }&\sum\limits_{i=1}^{\ell}a_{\ell i}a_{2i }&\cdots&\sum\limits_{i=1}^{\ell}a_{\ell i}^{2}\\
\end{pmatrix}.
$$  
Now by $k\mid q-1$, $\mathbb{F}_{q}^{*}=\langle\gamma\rangle$ and $\gamma^{\delta k}k+\sum\limits_{i=1}^{\ell}a_{1i}^{2}\in\mathbb{F}_{q}^{*}$, we can get $\rank\left(\boldsymbol{G}_{\boldsymbol{1},k}\boldsymbol{G}_{\boldsymbol{1},k}^{T}\right)=k$, thus by Lemma \ref{LCDequivalent}, the code $\mathrm{GRL}_{k}(\boldsymbol{\alpha}, \boldsymbol{v},\boldsymbol{A}_{\ell\times \ell})$ is Euclidean LCD.

This completes the proof of Theorem $\ref{EuclideanLCDRL1}$.

$\hfill\Box$

\begin{remark} 
	By taking $\delta=q-1$ and $\boldsymbol{A}_{l\times l}=\begin{pmatrix}
		0&1\\
		1&\tau
	\end{pmatrix}$ with $\tau\in\mathbb{F}_{q}$ in Theorem \ref{EuclideanLCDRL1} $(1)$, the corresponding result is just Lemma 3.5 $(1)$ in \cite{nonGRSLCD4}.
\end{remark}

By Lemma \ref{Ehull}, it's easy to obtain the following 
\begin{theorem}\label{ELCD1Hull1}
Let $\boldsymbol{\alpha}=\left(\gamma^{\delta }\alpha_{1},\ldots,\gamma^{\delta }\alpha_{k}\right)$, $\ell=\frac{k}{2}$ and $\gamma^{\delta k}k+\sum\limits_{i=1}^{\ell}a_{1i}^{2}=0$. Then $$\dim\left(\mathrm{Hull}_{E}\left(\mathrm{GRL}_{k}(\boldsymbol{\alpha}, \boldsymbol{v},\boldsymbol{A}_{\ell\times \ell})\right)\right)=1.$$
\end{theorem} 

The following Theorem \ref{EHull1conut} presents an upper bound for the number of the code $\mathrm{GRL}_{k}(\boldsymbol{\alpha}, \boldsymbol{v},\boldsymbol{A}_{\ell\times \ell})$ with $1$-dimensional hull in Theorem \ref{ELCD1Hull1}.
\begin{theorem}\label{EHull1conut}
If $\boldsymbol{\alpha}=\left(\gamma^{\delta }\alpha_{1},\ldots,\gamma^{\delta }\alpha_{k}\right)$, $\ell=\frac{k}{2}$ and $k+\sum\limits_{i=1}^{\ell}a_{1i}^{2}=0$. Then the following two statements are true,

$(1)$ if $\left(a_{11},\ldots,a_{1\ell}\right)\in\mathbb{F}_{q}^{\ell}\backslash\left\{\boldsymbol{0}\right\}$, then the number of the code $\mathrm{GRL}_{k}(\boldsymbol{\alpha}, \boldsymbol{v},\boldsymbol{A}_{\ell\times \ell})$ with $1$-dimensional hull is less than or equal to 
$$
\left(d\left(\frac{q-1}{2}\right)-1\right)  \cdot N_{f}\left(\ell,\gamma^{\delta k+\frac{q-1}{2}}k, q\right)\cdot\prod\limits_{i=1}^{\ell-1}\left(q^{\ell}-q^{i}\right);
$$

$(2)$ if $\left(a_{11},\ldots,a_{1\ell}\right)\in\left(\mathbb{F}_{q}^{*}\right)^{\ell}$, then the number of the code $\mathrm{GRL}_{k}(\boldsymbol{\alpha}, \boldsymbol{v},\boldsymbol{A}_{\ell\times \ell})$ with $1$-dimensional hull is less than or equal to 
$$\left(d\left(\frac{q-1}{2}\right)-1\right)\cdot N_{f}^{*}\left(\ell,\gamma^{\delta k+\frac{q-1}{2}}k, q\right)\cdot\prod\limits_{i=1}^{\ell-1}\left(q^{\ell}-q^{i}\right).$$
\end{theorem}
{\bf Proof.} By the equivalence of linear codes, it's easy to know that the number of the code $\mathrm{GRL}_{k}(\boldsymbol{\alpha}, \boldsymbol{v},\boldsymbol{A}_{\ell\times \ell})$ is less than or equal to the number of the matrix $\boldsymbol{G}_{\boldsymbol{1},k}$ given by $(\ref{GRL1generatormatrix})$. Furthermore, by Theorem \ref{ELCD1Hull1}, we know that the number of the code $\mathrm{GRL}_{k}(\boldsymbol{\alpha}, \boldsymbol{v},\boldsymbol{A}_{\ell\times \ell})$ with $1$-dimensional hull is less than or equal to the number of the matrix $\boldsymbol{G}_{\boldsymbol{1},k}$ satisfying the following three conditions simultaneously,

(i) $\boldsymbol{\alpha}=\left(\gamma^{\delta }\alpha_{1},\ldots,\gamma^{\delta }\alpha_{k}\right)$ with $\alpha_{i}=\gamma^{\frac{q-1}{k}i}(1\leq i\leq k)$;

(ii) $k=2\ell$;

(iii) $\boldsymbol{A}_{l\times l}=(a_{ij})_{\ell\times \ell}\in\mathrm{GL}_{\ell}\left(\mathbb{F}_{q}\right)$ with $\gamma^{\delta k}k+\sum\limits_{i=1}^{\ell}a_{1i}^{2}=0$. 

For the conditions (i) and (ii), it is easy to know that for given $q$, $\gamma$ and $\delta$, the number of the vector $\boldsymbol{\alpha}$ depends on the number $M$ of $k$, i.e., 
$$\begin{aligned}
M=&\#\left\{\boldsymbol{\alpha}=\left(\gamma^{\delta }\alpha_{1},\ldots,\gamma^{\delta }\alpha_{k}\right): \alpha_{i}=\gamma^{\frac{q-1}{k}i}, 1\leq i\leq k, k=2\ell\mid q-1\right\}\\
=&\#\left\{k:2\leq\ell\leq k\mid q-1\right\}\\
=&\#\left\{\ell:2\leq\ell\mid \frac{q-1}{2}\right\}\\
=&d\left(\frac{q-1}{2}\right)-1.
\end{aligned}$$

For the condition (iii), for the convenience, let $\boldsymbol{a}_{i}(1\leq i\leq \ell)$ be the $i$-row of $\boldsymbol{A}_{\ell\times\ell}=\left(a_{ij}\right)_{\ell\times\ell}$. By Lemma \ref{nonzerosFqsolution}, we know that the number of the $\ell$-tuples $\left(a_{11},\ldots,a_{1\ell}\right)\in\mathbb{F}_{q}^{\ell}\backslash\left\{\boldsymbol{0}\right\}$ such that $\gamma^{\delta k}k+\sum\limits_{i=1}^{\ell}a_{1i}^{2}=0$ is $N_{f}\left(\ell,\gamma^{\delta k+\frac{q-1}{2}}k, q\right)$, i.e., the vector $\boldsymbol{a}_{1}$ has $N_{f}\left(\ell,\gamma^{\delta k+\frac{q-1}{2}}k, q\right)$ choices. Note that $\boldsymbol{A}_{\ell\times\ell}=\left(a_{ij}\right)_{\ell\times\ell}\in\mathrm{GL}_{\ell}\left(\mathbb{F}_{q}\right)$ if and only if both $\boldsymbol{a}_{i}$ and $\boldsymbol{a}_{j}$ are $\mathbb{F}_{q}$-linearly independent for any $1\leq i\neq j\leq \ell$. Now for any given vector $\boldsymbol{a}_{1}$, the total number of the vector $\boldsymbol{a}_{2}$ which is $\mathbb{F}_{q}$-linearly independent of $\boldsymbol{a}_{1}$ is $q^{\ell}-q.$ Furthermore, for given two vectors $\boldsymbol{a}_{1}$ and $\boldsymbol{a}_{2}$, the total number of the vector $\boldsymbol{a}_{3}$ which is $\mathbb{F}_{q}$-linearly independent of both $\boldsymbol{a}_{1}$ and $\boldsymbol{a}_{2}$ is $q^{\ell}-q^{2}$. In the similar method as the above, it is easy to know that for given $i-1$ vectors $\boldsymbol{a}_{1},\ldots,\boldsymbol{a}_{i-1}$, the total number of the vector $\boldsymbol{a}_{i}$ which is $\mathbb{F}_{q}$-linearly independent of $\boldsymbol{a}_{1},\ldots,\boldsymbol{a}_{i-1}$ is $q^{\ell}-q^{i}$. And so, there are $$N_{f}\left(\ell,\gamma^{\delta k+\frac{q-1}{2}}k, q\right)\cdot\prod\limits_{i=1}^{\ell-1}\left(q^{\ell}-q^{i}\right)$$ choices for $\boldsymbol{A}_{\ell\times\ell}=\left(a_{ij}\right)_{\ell\times\ell}$ that satisfies the condition (iii). 

In summary of the above discussions, the number of $\boldsymbol{G}_{\boldsymbol{1},k}$ is less than or equal to
$$\left(d\left(\frac{q-1}{2}\right)-1\right)  \cdot N_{f}\left(\ell,\gamma^{\delta k+\frac{q-1}{2}}k, q\right)\cdot\prod\limits_{i=1}^{\ell-1}\left(q^{\ell}-q^{i}\right).$$

For the case of $\left(a_{11},\ldots,a_{1\ell}\right)\in\left(\mathbb{F}_{q}^{*}\right)^{\ell}$, it's easy to prove the corresponding result via the similar proofs presented above.

This completes the proof of Theorem $\ref{EHull1conut}$.
 
$\hfill\Box$
\subsection{The second class of LCD GRL codes with the parameters $\left[k+1+\ell,k\right]$}
In this subsection, by taking $\boldsymbol{\alpha}=\left(0,\gamma^{\delta}\alpha_{1},\ldots,\gamma^{\delta}\alpha_{k}\right)$, in the similar proofs as those for Theorems $\ref{EuclideanLCDRL1}$-$\ref{ELCD1Hull1}$, we construct two classes of Euclidean LCD GRL codes, and then get two classes of GRL codes with $1$-dimensional hull and a class of GRL codes with $2$-dimensional hull, i.e., we prove the following Theorems $\ref{EuclideanLCDRL2}$-\ref{EuclideanLCDRL21}.
\begin{theorem}\label{EuclideanLCDRL2}
	Let $\boldsymbol{\alpha}=\left(0,\gamma^{\delta}\alpha_{1}, \ldots, \gamma^{\delta}\alpha_{k}\right)$ with $\gcd(k+1,q)=1$ and $1\leq\delta\leq q-1$. Then the following two statements are true,
	
	$(1)$ if $\ell<\frac{k}{2}$, then the code $\mathrm{GRL}_{k}(\boldsymbol{\alpha}, \boldsymbol{v},\boldsymbol{A}_{\ell\times \ell})$ is Euclidean LCD;
	
	$(2)$ if $\ell=\frac{k}{2}$ and $\gamma^{\delta k}k+\sum\limits_{i=1}^{\ell}a_{1i}^{2}\in\mathbb{F}_{q}^{*}$, then the code $\mathrm{GRL}_{k}(\boldsymbol{\alpha}, \boldsymbol{v},\boldsymbol{A}_{\ell\times \ell})$ is Euclidean LCD.
\end{theorem}
\begin{remark} 
By taking $\delta=q-1$ and $\boldsymbol{A}_{l\times l}=\begin{pmatrix}
		0&1\\
		1&\tau
	\end{pmatrix}$ with $\tau\in\mathbb{F}_{q}$ in Theorem \ref{EuclideanLCDRL2} (1), the corresponding result is just Lemma 3.5 (2) in \cite{nonGRSLCD4}.
\end{remark}
\begin{theorem}\label{EuclideanLCDRL21}
Let $1\leq \delta\leq q-1$ and $\boldsymbol{\alpha}=\left(0,\gamma^{\delta}\alpha_{1}, \ldots, \gamma^{\delta}\alpha_{k}\right)$ with $p\mid k+1$. Then we have
$$
\dim\left(\mathrm{Hull}_{E}\left(\mathrm{GRL}_{k}(\boldsymbol{\alpha}, \boldsymbol{v},\boldsymbol{A}_{\ell\times \ell})\right)\right)=\begin{cases}
1,&\text{if}\ \ell<\frac{k}{2};\\
&\quad\text{or}\ \ell=\frac{k}{2}\ \text{and}\ \gamma^{\delta k}k+\sum\limits_{i=1}^{\ell}a_{1i}^{2}\in\mathbb{F}_{q}^{*};\\
2,&\text{if}\ \ell=\frac{k}{2}\ \text{and}\ \gamma^{\delta k}k+\sum\limits_{i=1}^{\ell}a_{1i}^{2}=0.\\
\end{cases}
$$
\end{theorem}
\subsection{The third class of LCD GRL codes with the parameters $\left[2k+\ell,k\right]$}
In this subsection, by taking $$\boldsymbol{\alpha}=\left(\gamma^{s}\alpha_{1},...,\gamma^{s}\alpha_{k},\gamma^{t}\alpha_{1},...,\gamma^{t}\alpha_{k}\right)$$
with $1\leq s\neq t\leq q-1$, we construct two classes of Euclidean LCD GRL codes, and then get a class of GRL codes with $1$-dimensional hull.

Firstly, we present the following key lemma. 

\begin{lemma}\label{kq-1sttaking}
	If $q-1$, $k$, $s$ and $t$ satisfy $\frac{q-1}{k}\nmid s-t$ and
	$$v_{2}\left(s-t\right)\neq v_{2}(q-1)-v_{2}(k)-1,$$ then any two components of the vector $\boldsymbol{\alpha}=\left(\gamma^{s}\alpha_{1},...,\gamma^{s}\alpha_{k},\gamma^{t}\alpha_{1},...,\gamma^{t}\alpha_{k}\right)$ are distinct over $\mathbb{F}_{q}$ and $\gamma^{sk}+\gamma^{t k}\in\mathbb{F}_{q}^{*}$. 
\end{lemma}	
{\bf Proof}. Firstly, by $\mathbb{F}_{q}^{*}=\langle\gamma\rangle$, we know that for any integers $m$ and $n$, $\gamma^{m}=\gamma^{n}$ if and only if $m\equiv n (\bmod\ q-1)$. And then for any $1\leq i\neq j\leq k\mid q-1$, $\gamma^{\frac{q-1}{k}i}\neq \gamma^{\frac{q-1}{k}j}$, i.e., $\alpha_{i}\neq \alpha_{j}.$ Note that $\gamma\in\mathbb{F}_{q}^{*}$, thus for any $1\leq i\neq j\leq k$, we have $\gamma^{s}\alpha_{i}\neq \gamma^{s}\alpha_{j}$ and $\gamma^{t}\alpha_{i}\neq \gamma^{t}\alpha_{j}$. Hence, we only need to prove that $\gamma^{s}\alpha_{i}\neq \gamma^{t}\alpha_{j}$ for any $s\neq t$ and $1\leq i,j\leq k$. In fact, by $\alpha_{i}\neq \alpha_{j}(1\leq i\leq k)$ and $\alpha_{i}=\gamma^{\frac{q-1}{k}i}(1\leq i\leq k)$, it's easy to know that $\gamma^{s}\alpha_{i}\neq \gamma^{t}\alpha_{j}$ if and only if for any $1\leq i,j\leq k$, $\gamma^{s}\alpha_{i}\neq \gamma^{t}\alpha_{j}$, i.e., for any $1-k\leq x\leq k-1$, 
$$x\frac{q-1}{k}\not\equiv s-t(\bmod\ q-1).$$
It's well-known that the binary linear Diophantine equation $x\frac{q-1}{k}\equiv s-t(\bmod\ q-1)$ is solvable on $\mathbb{F}_{q}$ if and only if $\gcd\left(\frac{q-1}{k},q-1\right)\mid s-t,$ i.e., $\frac{q-1}{k}\mid s-t$. And so, any two components of the vector $\boldsymbol{\alpha}=\left(\gamma^{s}\alpha_{1},...,\gamma^{s}\alpha_{k},\gamma^{t}\alpha_{1},...,\gamma^{t}\alpha_{k}\right)$ with  $\alpha_{i}=\gamma^{\frac{q-1}{k}i}$ are distinct over $\mathbb{F}_{q}$ if and only if $\frac{q-1}{k}\nmid s-t$.

Secondly, by $\mathbb{F}_{q}^{*}=\langle\gamma\rangle$, we have $\mathrm{ord}(\gamma)=q-1$, namely, $\gamma^{q-1}=1$, it means $\gamma^{\frac{q-1}{2}}=-1.$
Note that $\gamma^{sk}+\gamma^{t k}=0$ if and only if $\gamma^{sk}=-\gamma^{t k}=\gamma^{\frac{q-1}{2}+tk}$, i.e., $$sk\equiv \frac{q-1}{2}+tk(\bmod\ q-1).$$ 
For the convenience, we set $r=\frac{q-1}{2}$, i.e., $q-1=2r$. Thus $sk\equiv \frac{q-1}{2}+tk(\bmod\ q-1)$ if and only if there exists some $w\in\mathbb{Z}$ such that $ (s-t)k=(2w+1)r.$ It means $$v_{2}\left((s-t)k\right)=v_{2}\left((2w+1)r\right)=v_{2}(r)=v_{2}\left(\frac{q-1}{2}\right)=v_{2}(q-1)-1.$$
Note that $v_{2}((s-t)k)=v_{2}(s-t)+v_{2}(k)$, thus we know that if $(s-t)k\equiv \frac{q-1}{2}(\bmod\ q-1),$ then $v_{2}(s-t)+1=v_{2}(q-1)-v_{2}(k).$ 
And so, $(s-t) k\not\equiv \frac{q-1}{2}(\bmod\ q-1),$ i.e., $\gamma^{sk}+\gamma^{tk}\in\mathbb{F}_{q}^{*}$ if $v_{2}(s-t)+1\neq v_{2}(q-1)-v_{2}(k)$.

This completes the proof of Lemma $\ref{kq-1sttaking}$.

$\hfill\Box$

Especially, when $s=q-1$ and $t=\delta$ with $1\leq \delta \leq q-2$, i.e., $\boldsymbol{\alpha}=\left(\alpha_{1},...,\alpha_{k},\gamma^{\delta}\alpha_{1},...,\gamma^{\delta}\alpha_{k}\right)$, we can obtain the more precise result as the following 
\begin{lemma}\label{kq-1deltataking}
	If $q-1$, $k$, $\delta$ satisfy one of the following conditions, then any two components of the vector $\boldsymbol{\alpha}=\left(\alpha_{1},...,\alpha_{k},\gamma^{\delta}\alpha_{1},...,\gamma^{\delta}\alpha_{k}\right)$ are distinct over $\mathbb{F}_{q}$ and $1+\gamma^{\delta k}\in\mathbb{F}_{q}^{*}$.
	
	$(1)$ $v_{2}\left(k\right)=v_{2}\left(q-1\right)\geq 1$ and $\delta=2^{\mu}(\mu\geq 1)$;
	
	$(2)$ $v_{2}\left(k\right)<v_{2}\left(q-1\right)$ and $\delta=2^{v_{2}\left(q-1\right)-v_{2}\left(k\right)-\mu}(1<\mu\leq v_{2}\left(q-1\right)-v_{2}\left(k\right))$;
	
	$(3)$ $v_{2}\left(q-1\right)-v_{2}\left(k\right)\neq 1$, and there exists an odd prime $p^{\prime}$ such that $v_{p^{\prime}}\left(q-1\right)=v_{p^{\prime}}\left(k\right)$ and   $\delta=p_{i}^{v_{p^{\prime}}\left(q-1\right)+\mu}(\mu\geq 1)$;
	
	$(4)$ $\delta=2^{v_{2}\left(q-1\right)+\mu}(\mu\geq 0)$, and there exists an odd prime $p^{\prime}$ such that $v_{p^{\prime}}\left(k\right)<v_{p^{\prime}}\left(q-1\right)$;
	
	$(5)$ $v_{2}\left(q-1\right)-v_{2}\left(k\right)\neq 1$, and there exists an odd prime $p^{\prime}$ such that $v_{p^{\prime}}\left(k\right)<v_{p^{\prime}}\left(q-1\right)$ and $\delta=\left(p^{\prime}\right)^{v_{p^{\prime}}\left(q-1\right)-v_{p^{\prime}}\left(k\right)-\mu}(1\leq\mu\leq v_{p^{\prime}}\left(q-1\right)-v_{p^{\prime}}\left(k\right))$.
\end{lemma}

Based on the above Lemmas $\ref{kq-1sttaking}$-$\ref{kq-1deltataking}$, in the similar proofs as those for Theorems $\ref{EuclideanLCDRL1}$-\ref{ELCD1Hull1}, one can obtain the following Theorems \ref{EuclideanLCDRL31}-\ref{ELCD3Hull1}.
\begin{theorem}\label{EuclideanLCDRL31}
	Let $\gcd(2k,q)=1$, $\boldsymbol{\alpha}=\left(\gamma^{s}\alpha_{1},...,\gamma^{s}\alpha_{k},\gamma^{t}\alpha_{1},...,\gamma^{t}\alpha_{k}\right) \in \mathbb{F}_{q}^{2k}$ with $\frac{q-1}{k}\nmid s-t$ and
	$v_{2}\left(s-t\right)\neq v_{2}(q-1)-v_{2}(k)-1$. Then the code $\mathrm{GRL}_{k}(\boldsymbol{\alpha}, \boldsymbol{v},\boldsymbol{A}_{\ell\times \ell})$ is Euclidean LCD for $\ell<\frac{k}{2}$, or $\ell=\frac{k}{2}$ and $k\left(\gamma^{s k}+\gamma^{t k}\right)+\sum\limits_{i=1}^{\ell}a_{1i}^{2}\in\mathbb{F}_{q}^{*}$.
\end{theorem}

\begin{theorem}\label{EuclideanLCDRL3}
Let $\gcd(2k,q)=1$, $\boldsymbol{\alpha}=\left(\alpha_{1},...,\alpha_{k},\gamma^{\delta}\alpha_{1},...,\gamma^{\delta}\alpha_{k}\right) \in \mathbb{F}_{q}^{2k}$. If $k,q-1,\delta$ satisfy one of the following
	conditions, then the code $\mathrm{GRL}_{k}(\boldsymbol{\alpha}, \boldsymbol{v},\boldsymbol{A}_{\ell\times \ell})$ is Euclidean LCD for $\ell<\frac{k}{2}$, or $\ell=\frac{k}{2}$ and $k\left(1+\gamma^{\delta k}\right)+\sum\limits_{i=1}^{\ell}a_{1i}^{2}\in\mathbb{F}_{q}^{*}$.
	
$(1)$ $v_{2}\left(k\right)=v_{2}\left(q-1\right)\geq 1$ and $\delta=2^{\mu}(\mu\geq 1)$;

$(2)$ $v_{2}\left(k\right)<v_{2}\left(q-1\right)$ and $\delta=2^{v_{2}\left(q-1\right)-v_{2}\left(k\right)-\mu}(1<\mu\leq v_{2}\left(q-1\right)-v_{2}\left(k\right))$;

$(3)$ $v_{2}\left(q-1\right)-v_{2}\left(k\right)\neq 1$, there exists an odd  prime $p^{\prime}$ such that $v_{p^{\prime}}\left(q-1\right)=v_{p^{\prime}}\left(k\right)$, and   $\delta=\left(p^{\prime}\right)^{v_{p^{\prime}}\left(q-1\right)+\mu}(\mu\geq 1)$;

$(4)$ $\delta=2^{v_{2}\left(q-1\right)+\mu}(\mu\geq 0)$ and there exists an odd prime $p^{\prime}$ such that $v_{p^{\prime}}\left(k\right)<v_{p^{\prime}}\left(q-1\right)$;

$(5)$ $v_{2}\left(q-1\right)-v_{2}\left(k\right)\neq 1$, there exists an odd prime $p^{\prime}$ such that $v_{p^{\prime}}\left(k\right)<v_{p^{\prime}}\left(q-1\right)$, and $\delta=\left(p^{\prime}\right)^{v_{p^{\prime}}\left(q-1\right)-v_{p^{\prime}}\left(k\right)-\mu}(1\leq\mu\leq v_{p^{\prime}}\left(q-1\right)-v_{p^{\prime}}\left(k\right))$.
\end{theorem}
\begin{theorem}\label{ELCD3Hull1}
Let $\gcd\left(2k,q\right)=1$ and $\boldsymbol{\alpha}=\left(\alpha_{1},...,\alpha_{k},\gamma^{\delta}\alpha_{1},...,\gamma^{\delta}\alpha_{k}\right) \in \mathbb{F}_{q}^{2k}$. If $k,q-1,\delta$ satisfy one of the following
conditions, $\ell=\frac{k}{2}$ and $k\left(1+\gamma^{\delta k}\right)+\sum\limits_{i=1}^{\ell}a_{1i}^{2}=0$, then $$\dim\left(\mathrm{Hull}_{E}\left(\mathrm{GRL}_{k}(\boldsymbol{\alpha}, \boldsymbol{v},\boldsymbol{A}_{\ell\times \ell})\right)\right)=1.$$ 
	
$(1)$ $v_{2}\left(k\right)=v_{2}\left(q-1\right)\geq 1$ and $\delta=2^{\mu}(\mu\geq 1)$;

$(2)$ $v_{2}\left(k\right)<v_{2}\left(q-1\right)$ and $\delta=2^{v_{2}\left(q-1\right)-v_{2}\left(k\right)-\mu}(1<\mu\leq v_{2}\left(q-1\right)-v_{2}\left(k\right))$;

$(3)$ $v_{2}\left(q-1\right)-v_{2}\left(k\right)\neq 1$, there exists an odd  prime $p^{\prime}$ such that $v_{p^{\prime}}\left(q-1\right)=v_{p^{\prime}}\left(k\right)$, and   $\delta=\left(p^{\prime}\right)^{v_{p^{\prime}}\left(q-1\right)+\mu}(\mu\geq 1)$;

$(4)$ $\delta=2^{v_{2}\left(q-1\right)+\mu}(\mu\geq 0)$ and there exists an odd prime $p^{\prime}$ such that $v_{p^{\prime}}\left(k\right)<v_{p^{\prime}}\left(q-1\right)$;

$(5)$ $v_{2}\left(q-1\right)-v_{2}\left(k\right)\neq 1$, there exists an odd prime $p^{\prime}$ such that $v_{p^{\prime}}\left(k\right)<v_{p^{\prime}}\left(q-1\right)$, and $\delta=\left(p^{\prime}\right)^{v_{p^{\prime}}\left(q-1\right)-v_{p^{\prime}}\left(k\right)-\mu}(1\leq\mu\leq v_{p^{\prime}}\left(q-1\right)-v_{p^{\prime}}\left(k\right))$.
\end{theorem}
\begin{remark}\label{EuclideanLCDRL3Remark}
$(1)$ By taking $\delta=1$, $\mu=0$ and $\boldsymbol{A}_{l\times l}=\begin{pmatrix}
0&1\\
1&\tau
\end{pmatrix}$ with $\tau\in\mathbb{F}_{q}$ in Theorem \ref{EuclideanLCDRL3} (4), one can immediately get Lemma 3.5 (3) of the reference \cite{nonGRSLCD4}.

$(2)$ It's easy to know that in Theorems \ref{EuclideanLCDRL31}-\ref{ELCD3Hull1}, the condition $\gcd(2k,q)=1$ always holds  when $k\mid q-1$ and $q$ is an odd prime. Otherwise, if $p\mid k$, then by $k\mid q-1$, we have $p\mid q-1$, and so $p\mid -1$, which is a contradiction. 
\end{remark}
\subsection{Euclidean LCD GRL codes with the parameters $\left[3k+\ell,k\right]$}
In this subsection, by taking $$\boldsymbol{\alpha}=\left(\alpha_{1},...,\alpha_{k},\gamma\alpha_{1},...,\gamma\alpha_{k},\gamma^{2}\alpha_{1},...,\gamma^{2}\alpha_{k}\right),$$ 
we construct two classes of Euclidean LCD GRL codes, and then get three classes of GRL codes with $1$-dimensional hull and a class of GRL codes with $2$-dimensional hull.

\begin{theorem}\label{EuclideanLCDRL4}
	Let $\gcd(3k,q)=1$, $\boldsymbol{\alpha}=\left(\alpha_{1},...,\alpha_{k},\gamma\alpha_{1},...,\gamma\alpha_{k},\gamma^{2}\alpha_{1},...,\gamma^{2}\alpha_{k}\right) \in \mathbb{F}_{q}^{3k}$ and $q-1\notin\left\{k,2k,3k\right\}$. Then the code $\mathrm{GRL}_{k}(\boldsymbol{\alpha}, \boldsymbol{v},\boldsymbol{A}_{\ell\times \ell})$ is Euclidean LCD for $\ell<\frac{k}{2}$, or $\ell=\frac{k}{2}$ and $k\left(1+\gamma^{ k}+\gamma^{ 2k}\right)+\sum\limits_{i=1}^{\ell}a_{1i}^{2}\in\mathbb{F}_{q}^{*}$.
\end{theorem}
{\bf proof.} In the similar proof as that for Theorem $\ref{EuclideanLCDRL3}$, we only need to prove that the following two statements are true,

$(1)$ for any $1\leq i\neq j\leq k\mid q-1$, $\alpha_{i}\ne \gamma\alpha_{j},\alpha_{i}\ne \gamma^{2}\alpha_{j}$ and $\gamma\alpha_{i}\ne \gamma^{2}\alpha_{j};$

$(2)$ $k\left(1+\gamma^{k}+\gamma^{2k}\right)\in \mathbb{F}_{q}^{*}.$

{\bf For (1).} In the similar proofs as that for Lemma $\ref{kq-1sttaking}$, we know that for any $1\leq i\neq j\leq k\mid q-1$, $\alpha_{i}\ne \gamma\alpha_{j},\alpha_{i}\ne \gamma^{2}\alpha_{j}$ and $\gamma\alpha_{i}\ne \gamma^{2}\alpha_{j}$ if and only if 
$\frac{q-1}{k}\nmid 1$ and $\frac{q-1}{k}\nmid 2$,
i.e, the statement {\bf (1)} holds if and only if $q-1\notin \left\{k,2k\right\}$.

{\bf For (2).} Note that $2\leq k\mid\mathrm{ord}\left(\gamma\right)=q-1$, and so $\gamma^{k}-1\neq 0$ if and only if $k\neq q-1.$ 
By $\left(1+\gamma^{k}+\gamma^{2k}\right)\left(\gamma^{k}-1\right)=\gamma^{3k}-1$, we know that $1+\gamma^{k}+\gamma^{2k}\in \mathbb{F}_{q}^{*}$ if and only if $\gamma^{3k}-1\in \mathbb{F}_{q}^{*}$ and $k\neq q-1$, i.e., $\mathrm{ord}\left(\gamma\right)=q-1\nmid 3k$ and $k\neq q-1$, namely, $\frac{q-1}{k}\nmid 3$ and $k\neq q-1$, it means that the statement {\bf (2)} holds if and only if $ q-1\notin\left\{k,3k\right\}.$

From the above discussions, Theorem $\ref{EuclideanLCDRL4}$ is immediately.

$\hfill\Box$

In the similar analysis as that for Remark \ref{EuclideanLCDRL3Remark} $(2)$, it's easy to know that if $p\neq 3$ and $k\mid q-1$, then $\gcd(3k,q)=1$, and so we can get the following
\begin{theorem}\label{EuclideanLCDRL41}
	Let $p\geq 5$, $\boldsymbol{\alpha}=\left(\alpha_{1},...,\alpha_{k},\gamma\alpha_{1},...,\gamma\alpha_{k},\gamma^{2}\alpha_{1},...,\gamma^{2}\alpha_{k}\right) \in \mathbb{F}_{q}^{3k}$ and $q-1\notin\left\{k,2k,3k\right\}$. Then the code $\mathrm{GRL}_{k}(\boldsymbol{\alpha}, \boldsymbol{v},\boldsymbol{A}_{\ell\times \ell})$ is Euclidean LCD for $\ell<\frac{k}{2}$, or $\ell=\frac{k}{2}$ and $k\left(1+\gamma^{ k}+\gamma^{ 2k}\right)+\sum\limits_{i=1}^{\ell}a_{1i}^{2}\in\mathbb{F}_{q}^{*}$.
\end{theorem}

In the similar proof as that for Theorem \ref{ELCD1Hull1}, one can obtain the following
\begin{theorem}\label{ELCD4Hull11}
	Let $\boldsymbol{\alpha}=\left(\alpha_{1},...,\alpha_{k},\gamma\alpha_{1},...,\gamma\alpha_{k},\gamma^{2}\alpha_{1},...,\gamma^{2}\alpha_{k}\right) \in \mathbb{F}_{q}^{3k}$. Then for $q-1\notin\left\{k,2k,3k\right\}$, we have
$$\dim\left(\mathrm{Hull}_{E}\left(\mathrm{GRL}_{k}(\boldsymbol{\alpha}, \boldsymbol{v},\boldsymbol{A}_{\ell\times \ell})\right) \right)=\begin{cases}
1,& \text{if}\ p=3, \ell<\frac{k}{2},\\
&\quad\text{or}\ p=3, \ell=\frac{k}{2}\ \text{and}\ k\left(1+\gamma^{ k}+\gamma^{ 2k}\right)+\sum\limits_{i=1}^{\ell}a_{1i}^{2}\in\mathbb{F}_{q}^{*},\\
&\quad\text{or}\ \gcd(3k,q)=1, \ell=\frac{k}{2}\ \text{and}\ k\left(1+\gamma^{ k}+\gamma^{ 2k}\right)+\sum\limits_{i=1}^{\ell}a_{1i}^{2}=0;\\
2, &\text{if}\ p=3, \ell=\frac{k}{2}\ \text{and}\ k\left(1+\gamma^{ k}+\gamma^{ 2k}\right)+\sum\limits_{i=1}^{\ell}a_{1i}^{2}=0.
 \end{cases}$$ 
\end{theorem}
\subsection{Several classes of EAQECCs}
In this subsection, combining Theorems \ref{EuclideanLCDRL1}-\ref{EuclideanLCDRL4} and Lemma \ref{EAQECCsEhull}, we can immediately obtain several classes of EAQECCs as the following 
\begin{theorem}
Assume that $d$ is the minimum distance for the code $\mathrm{GRL}_{k}(\boldsymbol{\alpha}, \boldsymbol{v},\boldsymbol{A}_{\ell\times \ell})$. Then there exists some $q$-ary
EAQECCs with one of the following parameters,

$(1)$ $\left[\left[k+\ell,k-i,d,\ell-i\right]\right]_{q}$ for $i=0,1$; 

$(2)$ $\left[\left[k+1+\ell,k-i,d,\ell+1-i\right]\right]_{q}$ for $i=0,1,2$; 

$(3)$ $\left[\left[2k+\ell,k-i,d,k+\ell-i\right]\right]_{q}$ for $i=0,1$; 

$(4)$ $\left[\left[3k+\ell,k-i,d,2k+\ell-i\right]\right]_{q}$ for $i=0,1,2$.
\end{theorem}

In fact, for the given GRL code $\mathrm{GRL}_{k}(\boldsymbol{\alpha}, \boldsymbol{v},\boldsymbol{A}_{\ell\times \ell})$, its Euclidean dual code is uniquely determined. Furthermore, we also can immediately obtain several classes of ESQECCs as the following
\begin{theorem}
	Assume that $d^{\perp_{E}}$ is the minimum distance for the code $\mathrm{GRL}_{k}^{\perp_{E}}(\boldsymbol{\alpha}, \boldsymbol{v},\boldsymbol{A}_{\ell\times \ell})$. Then there exists some $q$-ary
	EAQECCs with one of the following parameters,
	
	$(1)$ $\left[\left[k+\ell,\ell-i,d^{\perp_{E}},k-i\right]\right]_{q}$ for $i=0,1$; 
	
	$(2)$ $\left[\left[k+1+\ell,\ell-i,d^{\perp_{E}},k+1-i\right]\right]_{q}$ for $i=0,1,2$; 
	
	$(3)$ $\left[\left[2k+\ell,k+\ell-i,d^{\perp_{E}},k-i\right]\right]_{q}$ for $i=0,1$; 
	
	$(4)$ $\left[\left[3k+\ell,2k+\ell-i,d^{\perp_{E}},k-i\right]\right]_{q}$ for $i=0,1,2$.
\end{theorem}
\section{Hermitian LCD GRL codes and their applications}\label{sec4} 
Throughout this section, we fix $k\mid q^2-1$, $\mathbb{F}_{q^2}^{*}=\mathbb{F}_{q^2}\backslash\left\{0\right\}=\langle\gamma\rangle$ and $\alpha_{i}=\gamma^{\frac{q^2-1}{k}i}(1\leq i\leq k)$.
In this section, by taking some special vector  $\boldsymbol{\alpha}=\left(\alpha_{1},\ldots,\alpha_{n}\right)$, we construct four classes of Hermitian LCD GRL codes with the parameters $\left[n+\ell,k\right]$, get several classes of GRL codes with small-dimensional hulls, and then for some GRL codes, we obtain an upper bound for the dimension of the hull. Finally, we obtain several families of EAQECCs. 

Firstly, by Lemma \ref{Fq*sum}, it's easy to obtain the following crucial lemma. 
\begin{lemma}\label{HLCDGRLlemma}
	Let $\boldsymbol{\alpha}=\left(\gamma^t\alpha_{1},\ldots,\gamma^t\alpha_{k}\right)$ with $1\leq t\leq q^2-1$, $\alpha_{i}=\gamma^{\frac{q^2-1}{k}i}$ and $\alpha_{i}\neq \alpha_{j}$ for $1\leq i\neq j\leq k$. Then there exists exactly one non-zero element over $\mathbb{F}_{q^2}$ for each row or each column of the matrix 
$$
\boldsymbol{M}_{k,t}=\begin{pmatrix}
k&\sum\limits_{i=1}^{k}\left(\gamma^t\alpha_{i}\right)^{q}&\sum\limits_{i=1}^{k}\left(\gamma^t\alpha_{i}\right)^{2q}&\cdots&\sum\limits_{i=1}^{k}\left(\gamma^t\alpha_{i}\right)^{\left(k-1\right)q}\\
	\sum\limits_{i=1}^{k}\left(\gamma^t\alpha_{i}\right)&\sum\limits_{i=1}^{k}\left(\gamma^t\alpha_{i}\right)^{1+q}&\sum\limits_{i=1}^{k}\left(\gamma^t\alpha_{i}\right)^{1+2q}&\cdots&\sum\limits_{i=1}^{k}\left(\gamma^t\alpha_{i}\right)^{1+\left(k-1\right)q}\\ 
	\vdots&\vdots&\vdots&&\vdots\\
	\sum\limits_{i=1}^{k}\left(\gamma^t\alpha_{i}\right)^{k-1}&\sum\limits_{i=1}^{k}\left(\gamma^t\alpha_{i}\right)^{k-1+q}&\sum\limits_{i=1}^{k}\left(\gamma^t\alpha_{i}\right)^{k-1+2q}&\cdots&\sum\limits_{i=1}^{k}\left(\gamma^t\alpha_{i}\right)^{k-1+\left(k-1\right)q}
\end{pmatrix}.
$$
\end{lemma}
{\bf Proof.} For any given $s(1\leq s\leq k)$, the $s$-column of $\boldsymbol{M}_{k,t}$ is $\begin{pmatrix}
	\sum\limits_{i=1}^{k}\left(\gamma^t\alpha_{i}\right)^{0+(s-1)q}\\
	\sum\limits_{i=1}^{k}\left(\gamma^t\alpha_{i}\right)^{1+(s-1)q}\\ 
	\vdots\\
	\sum\limits_{i=1}^{k}\left(\gamma^t\alpha_{i}\right)^{k-1+(s-1)q}\\
\end{pmatrix}$. Note that  $$\left\{0+sq,1+sq,\cdots,k-1+sq\right\}_{\bmod\ k}=\left\{0,1,\ldots,k-1\right\},$$
thus by Lemma \ref{Fq*sum}, every column of $\boldsymbol{M}_{k,t}$ has exactly one non-zero element. 

Now, for any given $s(1\leq s\leq k)$, the $s$-row of $\boldsymbol{M}_{k,t}$ is $$\begin{pmatrix}
	\sum\limits_{i=1}^{k}\left(\gamma^t\alpha_{i}\right)^{(s-1)+0\cdot q}&\sum\limits_{i=1}^{k}\left(\gamma^t\alpha_{i}\right)^{(s-1)+1\cdot q}&\cdots&\sum\limits_{i=1}^{k}\left(\gamma^t\alpha_{i}\right)^{(s-1)+(k-1)\cdot q}
\end{pmatrix}.$$ Note that 
$$(s-1)+0\cdot q,(s-1)+1\cdot q,\cdots,(s-1)+(k-1)q,$$
are $k$ positive integers, and for $0\leq i\neq j\leq k-1$, by $k\mid q^2-1$ and $q=p^m(m\in\mathbb{Z}^{+})$, we have $\gcd(p,k)=1$ and $1-k\leq i-j\leq k-1$, furthermore, it's easy to know that
$$
\begin{aligned}
	(s-1)+iq\equiv(s-1)+jq(\bmod\ k)\Longleftrightarrow&(i-j)q\equiv0(\bmod\ k)\\
	\Longleftrightarrow&k\mid (i-j)q\\
	\Longleftrightarrow&k\mid (i-j),
\end{aligned}
$$
which is a contradiction. And so, for $0\leq i\neq j\leq k-1$, we have
$$(s-1)+iq\not\equiv(s-1)+jq(\bmod\ k),$$
it means that $$\left\{(s-1)+0\cdot q,(s-1)+1\cdot q,\cdots,(s-1)+(k-1)q\right\}_{\bmod\ k}=\left\{0,1,\ldots,k-1\right\},$$
thus by Lemma \ref{Fq*sum}, every row of $\boldsymbol{M}_{k,t}$ has exactly one non-zero element.

From the above discussions, we complete the proof of Lemma $\ref{HLCDGRLlemma}$.

$\hfill\Box$

Especially, by taking $k\mid q-1$ or $k\mid q+1$ in Lemma \ref{HLCDGRLlemma}, we can obtain the following 
\begin{corollary}\label{Mk}
Let $\boldsymbol{\alpha}=\left(\gamma^t\alpha_{1},\ldots,\gamma^t\alpha_{k}\right)$ with $1\leq t\leq q^2-1$. Then the following two statements are true,

$(1)$ if $k\mid q-1$, then the $(i,k-i)$-th entry is non-zero of $\boldsymbol{M}_{k,t}$ for $1\leq i\leq k-1$, i.e., 
$$
\boldsymbol{M}_{k,t}=\begin{pmatrix}
k&0&0&\cdots&0\\
	0&0&0&\cdots&\left(\gamma^{t}\right)^{1+(k-1)q}k\\ 
	\vdots&\vdots&\vdots&&\vdots\\
	0&0&\left(\gamma^{t}\right)^{k-2+2q}k&\cdots&0\\
	0&\left(\gamma^{t}\right)^{k-1+q}k&0&\cdots&0
	\end{pmatrix};
$$

$(2)$ if $k\mid q+1$, then the $(i,i)$-th entry is non-zero of $\boldsymbol{M}_{k,t}$ for $2\leq i\leq k$, i.e., 
	$$
	\boldsymbol{M}_{k,t}=\begin{pmatrix}
		k&0&\cdots&0&0\\
		0&\left(\gamma^{t}\right)^{1+q}k&\cdots&0&0\\ 
		\vdots&\vdots&\ddots&\vdots&0\\
		0&0&\cdots&\left(\gamma^{t}\right)^{k-2+(k-2)q}k&0\\
		0&0&\cdots&0&\left(\gamma^{t}\right)^{k-1+(k-1)q}k
	\end{pmatrix}.
	$$
\end{corollary}
\subsection{Hermitian LCD GRL codes with the parameters $\left[k+\ell,k\right]$}
In this subsection, by taking 
$\boldsymbol{\alpha}=\left(\gamma^{\delta}\alpha_{1},\ldots,\gamma^{\delta}\alpha_{k}\right)$ with $1\leq \delta\leq q^2-1$, we construct two classes of Hermitian LCD GRL codes, and then get a class of GRL codes with $1$-dimensional hull. And for a class of GRL codes, we obtain an upper bound for the dimension of the hull. 
\begin{theorem}\label{HermitianLCDRL1}
	Let $\boldsymbol{\alpha}=\left(\gamma^{\delta}\alpha_{1},\ldots,\gamma^{\delta}\alpha_{k}\right)$ with $k\mid q-1$ and $1\leq \delta\leq q^2-1$. Then the following two statements are true,
	
	$(1)$ if $\ell<\frac{k}{2}$, then the code $\mathrm{GRL}_{k}(\boldsymbol{\alpha}, \boldsymbol{v},\boldsymbol{A}_{\ell\times \ell})$ is Hermitian LCD;
	
	$(2)$ if $\ell=\frac{k}{2}$ and $k\gamma^{\delta(k-\ell+\ell q)}+\sum\limits_{i=1}^{\ell}a_{1i}^{1+q}\in\mathbb{F}_{q^2}^{*}$, then the code $\mathrm{GRL}_{k}(\boldsymbol{\alpha}, \boldsymbol{v},\boldsymbol{A}_{\ell\times \ell})$ is Hermitian LCD.
\end{theorem}
{\bf Proof}. By Remark \ref{GRLmonomial}, we only focus on  $\mathrm{GRL}_{k}(\boldsymbol{\alpha}, \boldsymbol{1},\boldsymbol{A}_{l\times l})$ with the generator matrix $\boldsymbol{G}_{\boldsymbol{1},k}$ given by $(\ref{GRL1generatormatrix})$. Furthermore, by Lemma \ref{LCDequivalent}, we only need to prove that the $k\times k$ matrix $\boldsymbol{G}_{\boldsymbol{1},k}\boldsymbol{G}_{\boldsymbol{1},k}^{T}$ is nonsingular over $\mathbb{F}_{q^2}$, i.e.,   $\rank\left(\boldsymbol{G}_{\boldsymbol{1},k}\left(\overline{\boldsymbol{G}_{\boldsymbol{1},k}}\right)^T\right)=k$. In fact, let $\boldsymbol{L}_{k\times\ell}=\begin{pmatrix}
	\boldsymbol{0}_{\left(k-\ell\right)\times \ell}\\
	\boldsymbol{A}_{\ell\times\ell}
\end{pmatrix}$, then 
$$\boldsymbol{G}_{\boldsymbol{1},k}\left(\overline{\boldsymbol{G}_{\boldsymbol{1},k}}\right)^T=\boldsymbol{M}_{k,\delta}+\boldsymbol{L}_{\ell}\left(\overline{\boldsymbol{L}_{\ell}}\right)^T.
$$ 
Now by Corollary \ref{Mk}, we have
$$\small\begin{aligned}
&\boldsymbol{G}_{\boldsymbol{1}}\left(\overline{\boldsymbol{G}_{\boldsymbol{1}}}\right)^T\\
=&\begin{pmatrix}
	k&0&\cdots&0&0&\cdots&0&0&\cdots&0&0&\cdots&0\\
	0&0&\cdots&0&0&\cdots&0&0&\cdots&0&0&\cdots&\varUpsilon_{1}\\
	\vdots&\vdots& &\vdots&\vdots& &\vdots&\vdots& &\vdots&\vdots&\iddots &\vdots\\ 
	0&0&\cdots&0&0&\cdots&0&0&\cdots&0&\varUpsilon_{\ell}&\cdots&0\\ 
	0&0&\cdots&0&0&\cdots&0&0&\cdots&\varUpsilon_{\ell+1}&0&\cdots&0\\
	\vdots&\vdots& &\vdots&\vdots& &\vdots&\vdots& \iddots&\vdots&\vdots& &\vdots\\  
	0&0&\cdots&0&0&\cdots&0&\iddots&\cdots&0&0&\cdots&0\\ 
	0&0&\cdots&0&0&\cdots&\iddots&0&\cdots&0&0&\cdots&0\\
	\vdots&\vdots& &\vdots&\vdots& \iddots&\vdots&\vdots& &\vdots&\vdots& &\vdots\\
	0&0&\cdots&0&\varUpsilon_{k-\ell-1}&\cdots&0&0&\cdots&0&0&\cdots&0\\
	0&0&\cdots&\varUpsilon_{k-\ell}&0&\cdots&0&0&\cdots&0&\sum\limits_{i=1}^{\ell}a_{1i}^{1+q}&\cdots&\sum\limits_{i=1}^{\ell}a_{1i}a_{\ell i}^{q}\\
	\vdots&\vdots&\iddots &\vdots&\vdots& &\vdots&\vdots& &\vdots&\vdots& &\vdots\\
	0&\varUpsilon_{k-1}&\cdots&0&0&\cdots&0&0&\cdots&0&\sum\limits_{i=1}^{\ell}a_{\ell i}a_{1i }^{q}&\cdots&\sum\limits_{i=1}^{\ell}a_{\ell i}^{1+q}\\
\end{pmatrix},
\end{aligned}$$
where $\varUpsilon_{i}=k\gamma^{\delta(i+(k-i)q)}(1\leq i\leq k-1)$. Note that $2\leq k\mid q-1$ and $\mathbb{F}_{q^2}^{*}=\mathbb{F}_{q^2}\backslash\left\{0\right\}=\langle\gamma\rangle$, thus we have $k, \varUpsilon_{i}\in\mathbb{F}_{q^2}^{*}(1\leq i\leq k-1)$. Furthermore, by Lemma \ref{LCDequivalent}, the code $\mathrm{GRL}_{k}(\boldsymbol{\alpha}, \boldsymbol{v},\boldsymbol{A}_{\ell\times \ell})$ is Hermitian LCD when $\ell<\frac{k}{2}$ , or $\ell=\frac{k}{2}$ and $k\gamma^{\delta(k-\ell+\ell q)}+\sum\limits_{i=1}^{\ell}a_{1i}^{1+q}\in\mathbb{F}_{q^2}^{*}$.

From the above discussions, we complete the proof of Theorem $\ref{HermitianLCDRL1}$.

$\hfill\Box$

\begin{remark} 
	By taking $\delta=q^2-1$ and $\boldsymbol{A}_{l\times l}=\begin{pmatrix}
		0&1\\
		1&\tau
	\end{pmatrix}$ with $\tau\in\mathbb{F}_{q}^{2}$ in Theorem \ref{HermitianLCDRL1} $(1)$, one can get Lemma 3.12 $(1)$ of the reference  \cite{nonGRSLCD4}.
\end{remark}

By Lemma \ref{Ehull} and Corollary \ref{Mk} (2), it's easy to prove that the following two theorems.
\begin{theorem}\label{HLCD1Hull2}
	Let $\boldsymbol{\alpha}=\left(\gamma^{\delta}\alpha_{1},\ldots,\gamma^{\delta}\alpha_{k}\right)$ with $k\mid q+1$ and $1\leq \delta\leq q^2-1$. Then
	$$\dim\left(\mathrm{Hull}_{H}\left(\mathrm{GRL}_{k}(\boldsymbol{\alpha}, \boldsymbol{v},\boldsymbol{A}_{\ell\times \ell})\right)\right)\leq\ell.$$
\end{theorem}

\begin{remark}
By taking $\left(q,k,\ell,\delta\right)=\left(9,5,3,1\right)$ and $\boldsymbol{A}_{\ell\times \ell}=\begin{pmatrix}
\gamma^2&0&0\\
0&\gamma^3&0\\
0&0&\gamma^4\\
\end{pmatrix}$ in Theorem \ref{HLCD1Hull1}, and basing on the Magma program, we have $$\dim\left(\mathrm{Hull}_{H}\left(\mathrm{GRL}_{k}(\boldsymbol{\alpha}, \boldsymbol{v},\boldsymbol{A}_{\ell\times \ell})\right)\right)=3,$$
which means that the bound in Theorem \ref{HLCD1Hull1} is attainable.
\end{remark}

\begin{theorem}\label{HLCD1Hull1}
	Let $k\mid q-1$, $\boldsymbol{\alpha}=\left(\gamma^{\delta}\alpha_{1},\ldots,\gamma^{\delta}\alpha_{k}\right)$ with $1\leq \delta\leq q^2-1$. If $\ell=\frac{k}{2}$ and $k\gamma^{\delta(k-\ell+\ell q)}+\sum\limits_{i=1}^{\ell}a_{1i}^{1+q}=0$, then $\dim\left(\mathrm{Hull}_{H}\left(\mathrm{GRL}_{k}(\boldsymbol{\alpha}, \boldsymbol{v},\boldsymbol{A}_{\ell\times \ell})\right)\right)=1.$ 
\end{theorem} 

\subsection{Hermitian LCD GRL codes with the parameters $\left[k+1+\ell,k\right]$}
In this subsection, by taking 
$\boldsymbol{\alpha}=\left(0,\gamma^{\delta}\alpha_{1},\ldots,\gamma^{\delta}\alpha_{k}\right)$ with $1\leq \delta\leq q^2-1$, in the similar proofs as those for Theorems \ref{HermitianLCDRL1}-\ref{HLCD1Hull1}, we construct two classes of Hermitian LCD GRL codes, and then get two classes of GRL codes with $1$-dimensional hull and a class of GRL codes with $2$-dimensional hull. And for a class of GRL codes, we obtain an upper bound for the dimension of the hull   as the following Theorems \ref{HermitianLCDRL2}-\ref{HLCD2Hull2}. 
\begin{theorem}\label{HermitianLCDRL2}
	Let $k\mid q-1$, $\gcd(k+1,q)=1$ and  $\boldsymbol{\alpha}=\left(0,\gamma^{\delta}\alpha_{1}, \ldots, \gamma^{\delta}\alpha_{k}\right)$ with $1\leq \delta\leq q^2-1$. then the following two statements are true,
	
	$(1)$ if $\ell<\frac{k}{2}$, then the code $\mathrm{GRL}_{k}(\boldsymbol{\alpha}, \boldsymbol{v},\boldsymbol{A}_{\ell\times \ell})$ is Hermitian LCD;
	
	$(2)$ if $\ell=\frac{k}{2}$ and $k\gamma^{\delta(k-\ell+\ell q)}+\sum\limits_{i=1}^{\ell}a_{1i}^{1+q}\in\mathbb{F}_{q^2}^{*}$, then the code $\mathrm{GRL}_{k}(\boldsymbol{\alpha}, \boldsymbol{v},\boldsymbol{A}_{\ell\times \ell})$ is Hermitian LCD.
\end{theorem}

\begin{remark} 
	By taking $q^2-1\mid\delta$ and $\boldsymbol{A}_{l\times l}=\begin{pmatrix}
		0&1\\
		1&\tau
	\end{pmatrix}$ with $\tau\in\mathbb{F}_{q}^{2}$ in Theorem \ref{HermitianLCDRL2} $(1)$, one can get Lemma 3.12 $(2)$ of the reference \cite{nonGRSLCD4}.
\end{remark}
\begin{theorem}\label{HermitianLCDRL21}
	Let $p\mid k+1$ and  $\boldsymbol{\alpha}=\left(0,\gamma^{\delta}\alpha_{1}, \ldots, \gamma^{\delta}\alpha_{k}\right)$ with $1\leq \delta\leq q^2-1$. Then we have
$$
\dim\left(\mathrm{Hull}_{H}\left(\mathrm{GRL}_{k}(\boldsymbol{\alpha}, \boldsymbol{v},\boldsymbol{A}_{\ell\times \ell})\right)\right)=\begin{cases}
	1,&\text{if}\ \ell<\frac{k}{2};\\
	&\quad\text{or}\ \ell=\frac{k}{2}\ \text{and}\ k\gamma^{\delta(k-\ell+\ell q)}+\sum\limits_{i=1}^{\ell}a_{1i}^{1+q}\in\mathbb{F}_{q^2}^{*};\\
	2,&\text{if}\ \ell=\frac{k}{2}\ \text{and}\ k\gamma^{\delta(k-\ell+\ell q)}+\sum\limits_{i=1}^{\ell}a_{1i}^{1+q}=0.\\
\end{cases}
$$
\end{theorem}
\begin{theorem}\label{HLCD2Hull2}
	Let $\boldsymbol{\alpha}=\left(0,\gamma^{\delta}\alpha_{1},\ldots,\gamma^{\delta}\alpha_{k}\right)$ with $(k+1,q)=1$, $k\mid q+1$ and $1\leq \delta\leq q^2-1$. Then
	$$\dim\left(\mathrm{Hull}_{H}\left(\mathrm{GRL}_{k}(\boldsymbol{\alpha}, \boldsymbol{v},\boldsymbol{A}_{\ell\times \ell})\right)\right)\leq\ell.$$
\end{theorem}
\begin{remark}
By taking $\left(q,k,\ell,\delta\right)=\left(11,6,3,1\right)$ and $\boldsymbol{A}_{\ell\times \ell}=\begin{pmatrix}
		\gamma^4&0&0\\
		0&\gamma^5&0\\
		0&0&\gamma^6\\
	\end{pmatrix}$ in Theorem \ref{HLCD2Hull2}, and basing on the Magma program, we have $$\dim\left(\mathrm{Hull}_{H}\left(\mathrm{GRL}_{k}(\boldsymbol{\alpha}, \boldsymbol{v},\boldsymbol{A}_{\ell\times \ell})\right)\right)=3,$$
which means that the bound in Theorem \ref{HLCD2Hull2} is attainable.
\end{remark}
\subsection{Hermitian LCD GRL codes with the parameters $\left[2k+\ell,k\right]$}
In this subsection, by taking $$\boldsymbol{\alpha}=\left(\gamma^{s}\alpha_{1},...,\gamma^{s}\alpha_{k},\gamma^{t}\alpha_{1},...,\gamma^{t}\alpha_{k}\right)$$
with $1\leq s\neq t\leq q^2-1$, in the similar proof as that for Theorem \ref{EuclideanLCDRL31}, we construct two classes of Hermitain LCD GRL codes, and for a class of GRL codes, we obtain an upper bound for the dimension of the hull .

Firstly, in the similar proofs as those for Lemmas \ref{kq-1sttaking}-\ref{kq-1deltataking}, we can obtain the following key lemmas and corollarys.

\begin{lemma}\label{kq2-1sttaking}
Any two components of the vector $\boldsymbol{\alpha}=\left(\gamma^{s}\alpha_{1},...,\gamma^{s}\alpha_{k},\gamma^{t}\alpha_{1},...,\gamma^{t}\alpha_{k}\right)$ with $1\leq s\neq t\leq q^2-1$ are distinct if and only if $\frac{q^2-1}{k}\nmid s-t.$
\end{lemma}

\begin{corollary}\label{kq2-1deltataking}
Any two components of the vector $\boldsymbol{\alpha}=\left(\alpha_{1},...,\alpha_{k},\gamma^{\delta}\alpha_{1},...,\gamma^{\delta}\alpha_{k}\right)$ with $1\leq \delta\leq q^2-1$ are distinct if and only if $\frac{q^2-1}{k}\nmid\delta.$
\end{corollary} 
\begin{lemma}\label{kq2-1st}
	Let $$T=\left\{v_{2}(q^2-1)-1-v_{2}(i+(k-i)q)|1\leq i\leq k-1\right\}.$$ Then for any $1\leq i\leq k-1$, $\gamma^{s\left(i+(k-i)q\right)}+\gamma^{t\left(i+(k-i)q\right)}\in\mathbb{F}_{q^2}^{*}$  if and only if $(s-t)\left(i+(k-i)q\right)\not\equiv \frac{q^2-1}{2}(\bmod\ q^2-1)$. Especially, for $v_{2}(s-t)\notin T$, we have $\gamma^{s\left(i+(k-i)q\right)}+\gamma^{t\left(i+(k-i)q\right)}\in\mathbb{F}_{q^2}^{*}$ for any $1\leq i\leq k-1$.
\end{lemma}

\begin{corollary}\label{kq2-1delta}
	Let $$T=\left\{v_{2}(q^2-1)-1-v_{2}(i+(k-i)q)|1\leq i\leq k-1\right\}.$$ Then for any $1\leq i\leq k-1$, $1+\gamma^{\delta\left(i+(k-i)q\right)}\in\mathbb{F}_{q^2}^{*}$ if and only if $\delta\left(i+(k-i)q\right)\not\equiv \frac{q^2-1}{2}(\bmod\ q^2-1)$. Especially, for $v_{2}(\delta)\notin T$, we have $1+\gamma^{\delta\left(i+(k-i)q\right)}\in\mathbb{F}_{q^2}^{*}$ for any $1\leq i\leq k-1$.
\end{corollary}

Basing on the above Lemmas $\ref{kq2-1sttaking}$-$\ref{kq2-1delta}$, we can obtain the following Theorems \ref{HermitianLCDRL3}-\ref{HermitianLCDRL31}.
\begin{theorem}\label{HermitianLCDRL3}
	Let
	$$T=\left\{v_{2}(q^2-1)-1-v_{2}(i+(k-i)q)|1\leq i\leq k-1\right\},$$ and $\boldsymbol{\alpha}=\left(\gamma^{s}\alpha_{1},...,\gamma^{s}\alpha_{k},\gamma^{t}\alpha_{1},...,\gamma^{t}\alpha_{k}\right)\in \mathbb{F}_{q^2}^{2k}$ with $1\leq s\neq t\leq q^2-1$. If $v_{2}(s-t)\notin T$, $\frac{q^2-1}{k}\nmid s-t,$ $\gcd(2k,q)=1$ and $k\mid q-1$, then the code $\mathrm{GRL}_{k}(\boldsymbol{\alpha}, \boldsymbol{v},\boldsymbol{A}_{\ell\times \ell})$ is Hermitian LCD for $\ell<\frac{k}{2}$ or $\ell=\frac{k}{2}$ and $k\left(\gamma^{s\left(k-\ell+\ell q\right)}+\gamma^{t\left(k-\ell+\ell q\right)}\right)+\sum\limits_{i=1}^{\ell}a_{1i}^{1+q}\in\mathbb{F}_{q^2}^{*}$.
\end{theorem}
\begin{theorem}\label{HermitianLCDRL31}
	Let
	$$T=\left\{v_{2}(q^2-1)-1-v_{2}(i+(k-i)q)|1\leq i\leq k-1\right\},$$ and $\boldsymbol{\alpha}=\left(\alpha_{1},...,\alpha_{k},\gamma^{\delta}\alpha_{1},...,\gamma^{\delta}\alpha_{k}\right) \in \mathbb{F}_{q^2}^{2k}$ with $1\leq \delta\leq q^2-1$. If $v_{2}(\delta)\notin T$, $\frac{q^2-1}{k}\nmid\delta,$ $\gcd(2k,q)=1$ and $k\mid q-1$, then the code $\mathrm{GRL}_{k}(\boldsymbol{\alpha}, \boldsymbol{v},\boldsymbol{A}_{\ell\times \ell})$ is Hermitian LCD for $\ell<\frac{k}{2}$ or $\ell=\frac{k}{2}$ and $k\left(1+\gamma^{\delta\left(k-\ell+\ell q\right)}\right)+\sum\limits_{i=1}^{\ell}a_{1i}^{1+q}\in\mathbb{F}_{q^2}^{*}$.
\end{theorem}

Especially, when $\delta=2^{v_{2}\left(q^2-1\right)+\mu}(\mu\geq 0)$, it's easy to get $v_{2}(\delta)\notin T$. Furthermore, we can obtain the following 
\begin{theorem}\label{HermitianLCDRL32}
	Let $\boldsymbol{\alpha}=\left(\alpha_{1},...,\alpha_{k},\gamma^{\delta}\alpha_{1},...,\gamma^{\delta}\alpha_{k}\right) \in \mathbb{F}_{q^2}^{2k}$ with $1\leq \delta\leq q^2-1$ and $\delta=2^{v_{2}\left(q^2-1\right)+\mu}(\mu\geq 0)$. If   $\gcd(2k,q)=1$, $k\mid q-1$, and there exists an odd prime $p^{\prime}$ such that $v_{p^{\prime}}\left(k\right)<v_{p^{\prime}}\left(q^2-1\right)$, then the code $\mathrm{GRL}_{k}(\boldsymbol{\alpha}, \boldsymbol{v},\boldsymbol{A}_{\ell\times \ell})$ is Hermitian LCD for $\ell<\frac{k}{2}$ or $\ell=\frac{k}{2}$ and $k\left(1+\gamma^{\delta\left(k-\ell+\ell q\right)}\right)+\sum\limits_{i=1}^{\ell}a_{1i}^{1+q}\in\mathbb{F}_{q^2}^{*}$.
\end{theorem}

\begin{remark} 
	By taking $\mu=0$ and $\boldsymbol{A}_{l\times l}=\begin{pmatrix}
		0&1\\
		1&\tau
	\end{pmatrix}$ with $\tau\in\mathbb{F}_{q}^{2}$ in Theorem \ref{HermitianLCDRL32}, the corresponding result is just Lemma 3.12 $(3)$ in \cite{nonGRSLCD4}.
\end{remark}

Similarly, if there exists an odd prime $p^{\prime}$ such that $v_{p^{\prime}}\left(k\right)=v_{p^{\prime}}\left(q^2-1\right)$ and $\delta=\left(p^{\prime}\right)^{\mu}(\mu\geq 1)$, we have $v_{2}(\delta)=0$. Furthermore, we can obtain the following 
\begin{theorem}\label{HermitianLCDRL33}
	Let $\boldsymbol{\alpha}=\left(\alpha_{1},...,\alpha_{k},\gamma^{\delta}\alpha_{1},...,\gamma^{\delta}\alpha_{k}\right) \in \mathbb{F}_{q^2}^{2k}$ with $1\leq \delta\leq q^2-1$. If $\gcd(2k,q)=1$, $k\mid q-1$, $0\notin T$, and there exists an odd prime $p^{\prime}$ such that $v_{p^{\prime}}\left(k\right)=v_{p^{\prime}}\left(q^2-1\right)$ and $\delta=\left(p^{\prime}\right)^{\mu}(\mu\geq 1)$, then the code $\mathrm{GRL}_{k}(\boldsymbol{\alpha}, \boldsymbol{v},\boldsymbol{A}_{\ell\times \ell})$ is Hermitian LCD for $\ell<\frac{k}{2}$ or $\ell=\frac{k}{2}$ and $k\left(1+\gamma^{\delta\left(k-\ell+\ell q\right)}\right)+\sum\limits_{i=1}^{\ell}a_{1i}^{1+q}\in\mathbb{F}_{q^2}^{*}$.
\end{theorem}

In the similar proofs as those for Theorem \ref{HLCD1Hull1} and Theorem \ref{HermitianLCDRL3}, one can obtain the following
\begin{theorem}\label{HLCD3Hull1}
Let
$$N_{\ell}=\left\{v_{2}(q-1)-v_{2}(i)-1|1\leq i\leq k-\ell-1\right\},$$ and $\boldsymbol{\alpha}=\left(\gamma^{s}\alpha_{1},...,\gamma^{s}\alpha_{k},\gamma^{t}\alpha_{1},...,\gamma^{t}\alpha_{k}\right)\in \mathbb{F}_{q^2}^{2k}$ with $k\mid q+1$ and $1\leq s\neq t \leq q^2-1$.  If $v_{2}(s-t)\notin N_{\ell}$ and $\frac{q^2-1}{k}\nmid s-t$, then
	$$\dim\left(\mathrm{Hull}_{H}\left(\mathrm{GRL}_{k}(\boldsymbol{\alpha}, \boldsymbol{v},\boldsymbol{A}_{\ell\times \ell})\right)\right)\leq\ell.$$
\end{theorem}
\begin{remark}
By taking $\left(q,k,\ell,s,t\right)=\left(3^2,5,3,9,1\right)$ and $\boldsymbol{A}_{\ell\times \ell}=\begin{pmatrix}
		\gamma^6&0&0\\
		0&\gamma^7&0\\
		0&0&\gamma^8\\
	\end{pmatrix}$ in Theorem \ref{HLCD3Hull1}, and basing on the Magma program, we have $$\dim\left(\mathrm{Hull}_{H}\left(\mathrm{GRL}_{k}(\boldsymbol{\alpha}, \boldsymbol{v},\boldsymbol{A}_{\ell\times \ell})\right)\right)=3,$$
which means that the bound in Theorem \ref{HLCD3Hull1} is attainable.
\end{remark}
\subsection{Hermitian LCD GRL codes with the parameters $\left[(\delta+1)k+\ell,k\right]$}
In this subsection, by taking 
$$\boldsymbol{\alpha}=\left(\alpha_{1},...,\alpha_{k},\gamma\alpha_{1},...,\gamma\alpha_{k},\ldots,\gamma^{\delta}\alpha_{1},...,\gamma^{\delta}\alpha_{k}\right)$$
with $1\leq \delta\leq q$, we construct two classes of Hermitian LCD GRL codes, and then get two classes of GRL codes with $1$-dimensional hull and a class of GRL codes with $2$-dimensional hull. And for a class of GRL codes, we obtain an upper bound for the dimension of the hull.
\begin{theorem}\label{HermitianLCDRL4}
	Let $\Delta_{i}=i+(k-i)q(1\leq i\leq k-1)$ and $$\boldsymbol{\alpha}=\left(\alpha_{1},...,\alpha_{k},\gamma\alpha_{1},...,\gamma\alpha_{k},\ldots,\gamma^{\delta}\alpha_{1},...,\gamma^{\delta}\alpha_{k}\right)$$ with $1\leq\delta\leq q$. If  $\gcd\left(\left(\delta+1\right)k,q\right)=1$, $k\mid q-1$, and  $\frac{q^2-1}{\gcd\left(q^2-1,\Delta_{i}\right)}\nmid \delta+1$ for any $1\leq i\leq k-1$. Then the code $\mathrm{GRL}_{k}(\boldsymbol{\alpha}, \boldsymbol{v},\boldsymbol{A}_{\ell\times \ell})$ is Hermitian LCD for $\ell<\frac{k}{2}$ or $\ell=\frac{k}{2}$ and $k\sum\limits_{j=0}^{\delta}\gamma^{j\left(\ell+(k-\ell)q\right)}+\sum\limits_{i=1}^{\ell}a_{1i}^{1+q}\in\mathbb{F}_{q^2}^{*}$.
\end{theorem}
{\bf Proof}. In the similar proof as that for Theorem $\ref{HermitianLCDRL1}$, we only need to prove that the following two statements are true,

$(1)$ for any $1\leq i,j\leq k$ and $0\leq s,t\leq \delta$, $\gamma^{s}\alpha_{i}\neq \gamma^{t}\alpha_{j};$ 

$(2)$ for any $1\leq i\leq k-1$, $\sum\limits_{j=0}^{\delta}\gamma^{j(i+(k-i)q)}\in\mathbb{F}_{q^2}^{*}.$ 

{\bf For (1).} In the similar proof as that for Lemma \ref{kq-1sttaking}, we know that the statements {\bf (1)} holds if and only if $\frac{q^2-1}{k}\nmid t$ for any $-\delta\leq t\leq \delta$. By $q+1\mid q^2-1$ and $k\mid q-1$, we know that if there exists some $-\delta\leq t\leq \delta$ such that $\frac{q^2-1}{k}\mid t$, then $q+1\mid t$, which is contradict with $t\leq \delta\leq q$.

{\bf For (2).} By $k\mid q-1$, $1\leq i \leq k-1$ and $\Delta_{i}=kq+(1-q)i$, we have $$-(q^2-4q+2)\leq\Delta_{i}\leq q^2-2q+1.$$
Since $q$ is an odd prime power, thus for any $1\leq i\leq k-1$, $$-(q^2-1)<\Delta_{i}<q^2-1,$$ i.e., $\gamma^{\Delta_{i}}-1\neq 0$ for any $1\leq i\leq k-1$. Note that
$$\left(\gamma^{\Delta_{i}}-1\right)\sum\limits_{j=0}^{\delta}\gamma^{j\Delta_{i}}=\gamma^{(\delta+1)\Delta_{i}}-1,$$
then $\sum\limits_{j=0}^{\delta}\gamma^{j\Delta_{i}}\in \mathbb{F}_{q^2}^{*}$ if and only if $\gamma^{(\delta+1)\Delta_{i}}-1\in \mathbb{F}_{q^2}^{*}$, i.e., $\mathrm{ord}\left(\gamma\right)=q^2-1\nmid (\delta+1)\Delta_{i}$, namely, $$\frac{q^2-1}{\gcd\left(q^2-1,\Delta_{i}\right)}\nmid \delta+1.$$

From the above discussions, we complete the proof of Theorem $\ref{HermitianLCDRL4}$.

$\hfill\Box$

In the similar proof as that for Theorem \ref{ELCD1Hull1}, one can obtain the following
\begin{theorem}\label{HermitianLCDRL4Hull}
	Let $\Delta_{i}=i+(k-i)q(1\leq i\leq k-1)$,  $$S_{t}=\left\{i\middle|\frac{q^2-1}{\gcd\left(q^2-1,\Delta_{i}\right)}\mid \delta+1,t\leq i\leq k-1\right\}$$
and $$\boldsymbol{\alpha}=\left(\alpha_{1},...,\alpha_{k},\gamma\alpha_{1},...,\gamma\alpha_{k},\ldots,\gamma^{\delta}\alpha_{1},...,\gamma^{\delta}\alpha_{k}\right)$$ with $1\leq \delta\leq q$, $k\mid q-1$. Then we have
$$
\begin{aligned}
&\dim\left(\mathrm{Hull}_{H}\left(\mathrm{GRL}_{k}(\boldsymbol{\alpha}, \boldsymbol{v},\boldsymbol{A}_{\ell\times \ell})\right)\right)\\
=&\begin{cases}
	1,&\ \text{if}\ \# S_{1}=0, p\mid \delta+1\ \text{and}\ \ell<\frac{k}{2};\\
	&\quad\text{or}\ \# S_{1}=0, p\mid \delta+1,\ell=\frac{k}{2},\ \text{and}\ k\sum\limits_{i=0}^{\delta}\gamma^{\delta\left(\ell+(k-\ell)q\right)}+\sum\limits_{i=1}^{\ell}a_{1i}^{1+q}\in\mathbb{F}_{q^2}^{*};\\
	&\quad\text{or}\ \# S_{1}=0,  \gcd\left(k\left(\delta+1\right),q\right)=1,\ell=\frac{k}{2},\ \text{and}\ k\sum\limits_{i=0}^{\delta}\gamma^{\delta\left(\ell+(k-\ell)q\right)}+\sum\limits_{i=1}^{\ell}a_{1i}^{1+q}=0;\\
	2,&\ \text{if}\ \# S_{1}=0, p\mid \delta+1,\ell=\frac{k}{2},\ \text{and}\ k\sum\limits_{i=0}^{\delta}\gamma^{\delta\left(\ell+(k-\ell)q\right)}+\sum\limits_{i=1}^{\ell}a_{1i}^{1+q}=0;\\
s,&\ \text{if}\ \# S_{k-\ell}=s, \gcd\left(k\left(\delta+1\right),q\right)=1,\ell<\frac{k}{2};\\
s+1,&\ \text{if}\ \# S_{k-\ell}=s, p\mid \delta+1,\ell<\frac{k}{2}.
\end{cases}
\end{aligned}
$$ 
\end{theorem}
\begin{theorem}\label{HLCD4Hull1}
	Let
	$\boldsymbol{\alpha}=\left(\alpha_{1},...,\alpha_{k},\gamma\alpha_{1},...,\gamma\alpha_{k},\ldots,\gamma^{\delta}\alpha_{1},...,\gamma^{\delta}\alpha_{k}\right)$ with $1\leq \delta\leq q$, $k\mid q+1$, $\varLambda_{i}=i(q+1)(1\leq i\leq k-1)$ and 
	$$U_{\ell}=\left\{i\middle|\frac{q^2-1}{\gcd\left(q^2-1,\varLambda_{i}\right)}\mid \delta+1,1\leq i\leq k-\ell-1\right\}.$$ If $\gcd\left((\delta+1)k,q\right)=1$ and $\# U_{\ell}=0$, then
	$$\dim\left(\mathrm{Hull}_{H}\left(\mathrm{GRL}_{k}(\boldsymbol{\alpha}, \boldsymbol{v},\boldsymbol{A}_{\ell\times \ell})\right)\right)\leq\ell.$$
\end{theorem}
\begin{remark}
By taking $\left(q,k,\ell,\delta\right)=\left(3^2,5,3,3\right)$ and $\boldsymbol{A}_{\ell\times \ell}=\begin{pmatrix}
		\gamma^8&0&0\\
		0&\gamma^5&0\\
		0&0&\gamma^8\\
	\end{pmatrix}$ in Theorem \ref{HLCD4Hull1}, and basing on the Magma program, we have $$\dim\left(\mathrm{Hull}_{H}\left(\mathrm{GRL}_{k}(\boldsymbol{\alpha}, \boldsymbol{v},\boldsymbol{A}_{\ell\times \ell})\right)\right)=3,$$
which means that the bound in Theorem \ref{HLCD4Hull1} is attainable.
\end{remark}
\subsection{Several classes of EAQECCs}
In this subsection, combining Lemma \ref{EAQECCsHhull} and Theorems \ref{HermitianLCDRL1},\ref{HLCD1Hull1}-\ref{HermitianLCDRL21},\ref{HermitianLCDRL3}-\ref{HermitianLCDRL33}, \ref{HermitianLCDRL4}-\ref{HermitianLCDRL4Hull}, we can immediately obtain several classes of ESQECCs as the following 
\begin{theorem}
Assume that $d$ is the minimum distance for the code $\mathrm{GRL}_{k}(\boldsymbol{\alpha}, \boldsymbol{v},\boldsymbol{A}_{\ell\times \ell})$. Then there exists some q-ary
EAQECCs with  one of the following parameters,

$(1)$ $\left[\left[k+\ell,k-i,d,\ell-i\right]\right]_{q}$ for $i=0,1$; 

$(2)$ $\left[\left[k+1+\ell,k-i,d,\ell+1-i\right]\right]_{q}$ for $i=0,1,2$; 

$(3)$ $\left[\left[2k+\ell,k,d,k+\ell\right]\right]_{q}$;

$(4)$ $\left[\left[\left(\delta+1\right)k+\ell,k-i,d,\delta k+\ell-i\right]\right]_{q}$  for  $i=0,1,2,\ldots,s,s+1$ with $s\leq \ell$.
\end{theorem}

In fact, for the given GRL code $\mathrm{GRL}_{k}(\boldsymbol{\alpha}, \boldsymbol{v},\boldsymbol{A}_{\ell\times \ell})$, its Hermitian dual code is uniquely determined. Furthermore, we also can immediately obtain several classes of ESQECCs as the following
\begin{theorem}
	Assume that $d^{\perp_H}$ is the minimum distance for the code $\mathrm{GRL}_{k}^{\perp_{H}}(\boldsymbol{\alpha}, \boldsymbol{v},\boldsymbol{A}_{\ell\times \ell})$. Then there exists some q-ary
	EAQECCs with  one of the following parameters,
	
	$(1)$ $\left[\left[k+\ell,\ell-i,d^{\perp_H},k-i\right]\right]_{q}$ for $i=0,1$;
	
	$(2)$ $\left[\left[k+1+\ell,\ell+1-i,d^{\perp_H},k-i\right]\right]_{q}$ for $i=0,1,2$; 
	
	$(3)$ $\left[\left[2k+\ell,k+\ell,d^{\perp_{H}},k\right]\right]_{q}$; 
	
	$(4)$ $\left[\left[\left(\delta+1\right)k+\ell,\delta k+\ell-i,d^{\perp_{H}},k-i\right]\right]_{q}$ for  $i=0,1,2,\ldots,s,s+1$ with $s\leq \ell$.
\end{theorem}
\section{The non-RS property of the GRL code}\label{sec5}
In this section, by using the Cauchy matrix method proposed in \cite{Roth1985}, we study the non-GRS property of the code $\mathrm{GRL}_{k}(\boldsymbol{\alpha},\boldsymbol{v},\boldsymbol{A}_{\ell\times \ell})$, and prove that the code $\mathrm{GRL}_{k}(\boldsymbol{\alpha},\boldsymbol{v},\boldsymbol{A}_{\ell\times \ell})$ is non-GRS for $k>\ell$. And when $k=\ell$, we give some examples.

 Firstly, we present the following crucial lemma.
\begin{lemma}\label{fi(x)}
Let $\sigma_{i}$ be the $i$-th degree elementary symmetric polynomial and $$
f(x)=\prod\limits_{i=1}^{k}\left(x-\alpha_{i}\right)=\sum\limits_{i=0}^{k} (-1)^{i}\sigma_{i} x^{k-i}.
$$ Then for $$
f_{i}(x)=\sum\limits_{j=1}^{k} f_{i j} x^{j-1}=\prod\limits_{j=1, j \neq i}^{k}\left(x-\alpha_{j}\right)(1 \leq i \leq k)
,$$ we have
$$f_{ij}=\sum\limits_{s=0}^{k-j}(-1)^{s}\sigma_{s}\alpha_{i}^{k-j-s}(1\leq j\leq k).$$
\end{lemma}
\textbf{Proof}. Note that 
$$f(x)=\left(x-\alpha_{i}\right)f_{i}(x),$$
then by comparing the coefficients for $x^{i}$ in both sides, we have 
$$
\begin{cases}
f_{ik}=1,\\
f_{i(k-1)}-f_{ik}\alpha_{i}=-\sigma_{1},\\
f_{i(k-2)}-f_{i(k-1)}\alpha_{i}=\sigma_{2},\\
\ \ \ \ \ \ \ \ \ \vdots\\
f_{i1}-f_{i2}\alpha_{i}=(-1)^{k-1}\sigma_{k-1},\\
-f_{i1}\alpha_{i}=(-1)^{k}\sigma_{k}.
\end{cases}
$$
Furthermore, we can get 
$$\begin{cases}
f_{ik}=1,\\
f_{i(k-1)}=\alpha_{i}-\sigma_{1},\\
f_{i(k-2)}=\left(\alpha_{i}-\sigma_{1}\right)\alpha_{i}+\sigma_{2},\\
\ \ \ \ \ \ \ \ \ \vdots\\
f_{i2}=\sum\limits_{s=0}^{k-2}(-1)^{s}\sigma_{s}\alpha_{i}^{k-2-s},\\
f_{i1}=\sum\limits_{s=0}^{k-1}(-1)^{s}\sigma_{s}\alpha_{i}^{k-1-s}.
\end{cases}$$ 

This completes the proof of Lemma $\ref{fi(x)}$.

$\hfill\Box$

Now, we prove that the code $\mathrm{GRL}_{k}(\boldsymbol{\alpha},\boldsymbol{v},\boldsymbol{A}_{\ell\times \ell})$ is non-GRS when $k>\ell$. 
\begin{theorem}\label{nonRS}
	If $k>\ell$, then the code $\mathrm{GRL}_{k}(\boldsymbol{\alpha},\boldsymbol{v},\boldsymbol{A}_{\ell\times \ell})$ is non-GRS.
\end{theorem}
\textbf{Proof}. By Remark \ref{GRLmonomial}, we only focus on the code $\mathrm{GRL}_{k}(\boldsymbol{\alpha},\boldsymbol{1},\boldsymbol{A}_{\ell\times \ell})$. Firstly, we set $$
f_{i}(x)=\sum_{j=1}^{k} f_{i j} x^{j-1}=\prod_{j=1, j \neq i}^{k}\left(x-\alpha_{j}\right)(1 \leq i \leq k)
,$$ 
\begin{equation}\label{F}
	\boldsymbol{F}=\begin{pmatrix}
		f_{11} & f_{12} & \cdots & f_{1 k} \\
		f_{21} & f_{22} & \cdots & f_{2 k} \\
		\vdots & \vdots & & \vdots \\
		f_{k 1} & f_{k 2} & \cdots & f_{k k}
	\end{pmatrix}
\end{equation}
and
$$
\eta_{i}=\prod_{s=1, s \neq i}^{k}\left(\alpha_{i}-\alpha_{s}\right)(1\leq i\leq k),$$
$$ \eta_{k+j}=\prod_{s=1}^{k}\left(\alpha_{k+j}-\alpha_{s}\right)(1\leq j\leq n-k).
$$ 

Now for $\boldsymbol{F}$ and $\boldsymbol{G}_{1}$ given by $(\ref{F})$ and $(\ref{GRL1generatormatrix})$, respectively, we have 
$$\boldsymbol{FG_{1}}=\begin{pmatrix}
f_{1}\left(\alpha_{1}\right)&\cdots&f_{1}\left(\alpha_{k}\right)&f_{1}\left(\alpha_{k+1}\right)&\cdots&f_{1}\left(\alpha_{n}\right)&\sum\limits_{s=1}^{\ell}f_{1(k-\ell+s)}a_{s1}&\cdots&\sum\limits_{s=1}^{\ell}f_{1(k-\ell+s)}a_{s\ell}\\
\vdots& &\vdots&\vdots& &\vdots&\vdots& &\vdots\\
f_{k}\left(\alpha_{1}\right)&\cdots&f_{k}\left(\alpha_{k}\right)&f_{k}\left(\alpha_{k+1}\right)&\cdots&f_{k}\left(\alpha_{n}\right)&\sum\limits_{s=1}^{\ell}f_{k(k-\ell+s)}a_{s1}&\cdots&\sum\limits_{s=1}^{\ell}f_{k(k-\ell+s)}a_{s\ell}
\end{pmatrix}.$$
Note that
$$f_{i}\left(\alpha_{j}\right)=\begin{cases}
\eta_{j},&1\leq i=j\leq k;\\
0,&1\leq i\neq j\leq k;\\
\frac{\eta_{j}}{\alpha_{j}-\alpha_{i}},&k+1\leq j\leq n,
\end{cases}$$
thus
$$\begin{aligned}
&\boldsymbol{FG_{1}}\\
=&\begin{pmatrix}
	\eta_{1}&\cdots&0&\frac{\eta_{k+1}}{\alpha_{k+1}-\alpha_{1}}&\cdots&\frac{\eta_{n}}{\alpha_{n}-\alpha_{1}}&\sum\limits_{s=1}^{\ell}f_{1(k-\ell+s)}a_{s1}&\cdots&\sum\limits_{s=1}^{\ell}f_{1(k-\ell+s)}a_{s\ell}\\
	\vdots& &\vdots&\vdots& &\vdots&\vdots& &\vdots\\
	0&\cdots&\eta_{k}&\frac{\eta_{k+1}}{\alpha_{k+1}-\alpha_{k}}&\cdots&\frac{\eta_{n}}{\alpha_{n}-\alpha_{k}}&\sum\limits_{s=1}^{\ell}f_{k(k-\ell+s)}a_{s1}&\cdots&\sum\limits_{s=1}^{\ell}f_{k(k-\ell+s)}a_{s\ell}
\end{pmatrix}\\
=&\boldsymbol{}\begin{pmatrix}
\eta_{1}&\cdots&0\\
\vdots& &\vdots\\
0&\cdots&\eta_{k}
\end{pmatrix}\begin{pmatrix}
	1&\cdots&0&\frac{\eta_{k+1}\eta_{1}^{-1}}{\alpha_{k+1}-\alpha_{1}}&\cdots&\frac{\eta_{n}\eta_{1}^{-1}}{\alpha_{n}-\alpha_{1}}&\eta_{1}^{-1}\sum\limits_{s=1}^{\ell}f_{1(k-\ell+s)}a_{s1}&\cdots&\eta_{1}^{-1}\sum\limits_{s=1}^{\ell}f_{1(k-\ell+s)}a_{s\ell}\\
	\vdots& &\vdots&\vdots& &\vdots&\vdots& &\vdots\\
	0&\cdots&1&\frac{\eta_{k+1}\eta_{k}^{-1}}{\alpha_{k+1}-\alpha_{k}}&\cdots&\frac{\eta_{n}\eta_{k}^{-1}}{\alpha_{n}-\alpha_{k}}&\eta_{k}^{-1}\sum\limits_{s=1}^{\ell}f_{k(k-\ell+s)}a_{s1}&\cdots&\eta_{k}^{-1}\sum\limits_{s=1}^{\ell}f_{k(k-\ell+s)}a_{s\ell}
\end{pmatrix}\\
=&\boldsymbol{V}\widetilde{\boldsymbol{G}_{\boldsymbol{1}}}\\
=&\boldsymbol{V}\left[\boldsymbol{E}_{k}\mid\boldsymbol{B}\right].
\end{aligned}$$ 
It's easy to know that $\boldsymbol{F}$ and $\boldsymbol{V}$ are both nonsingular over $\mathbb{F}_{q}$, and so $\boldsymbol{FG}_{\boldsymbol{1}}$ and 
$\widetilde{\boldsymbol{G}_{\boldsymbol{1}}}$ are both the generator matrices of the code $\mathrm{GRL}_{k}(\boldsymbol{\alpha},\boldsymbol{1},\boldsymbol{A}_{\ell\times \ell}).$

Note that $\boldsymbol{A}_{\ell\times \ell}\in\mathrm{GL}_{\ell}\left(\mathbb{F}_{q}\right)$, it means that $a_{11},a_{21},\cdots,a_{\ell1}$ are not all equal to zero, and then without loss of generality, we can suppose that $a_{\ell1}\neq 0$. Thus, if $\widetilde{\boldsymbol{G}_{1}}$ generates a $\mathrm{RS}$ code, then by Lemma \ref{nonGRS}, for the $(i,n-k+1)$-th entry of $\boldsymbol{B}$($1\leq i\leq k$),  there exist $\alpha_{n+1},\cdots,\alpha_{n+\ell}\in\mathbb{F}_{q}\backslash\left\{\alpha_1,\ldots,\alpha_n\right\}$ such that
$$
\eta_{i}^{-1}\sum\limits_{s=1}^{\ell}f_{i(k-\ell+s)}a_{s1}=\frac{\eta_{n+1}\eta_{i}^{-1}}{\alpha_{n+1}-\alpha_{i}},$$
where $\eta_{n+1}=\prod\limits_{s=1}^{k}\left(\alpha_{n+1}-\alpha_{s}\right),$ 
i.e., 
$$\frac{\eta_{n+1}}{\alpha_{n+1}-\alpha_{i}}=\sum\limits_{s=1}^{\ell}f_{i(k-\ell+s)}a_{s1}=\sum\limits_{s=1}^{\ell}\left(\sum\limits_{t=0}^{\ell-s}(-1)^{t}\sigma_{t}\alpha_{i}^{\ell-s-t}\right)a_{s1}.$$ 
It's means that $\alpha_1,\ldots,\alpha_k(k>\ell)$ are distinct roots of the polynomial $$\eta_{n+1}=\left(\alpha_{n+1}-x\right)\sum\limits_{s=1}^{\ell}\left(\sum\limits_{t=0}^{\ell-s}(-1)^{t}\sigma_{t}x^{\ell-s-t}\right)a_{s1},$$ 
which is contradict with $\deg\left(\left(\alpha_{n+1}-x\right)\sum\limits_{s=1}^{\ell}\left(\sum\limits_{t=0}^{\ell-s}(-1)^{t}\sigma_{t}x^{\ell-s-t}\right)a_{s1}\right)=\ell$. Therefore $\widetilde{\boldsymbol{G}_{\boldsymbol{1}}}$ is not a generator matrix for any $\mathrm{RS}$ code, i.e., the code $\mathrm{GRL}_{k}(\boldsymbol{\alpha},\boldsymbol{1},\boldsymbol{A}_{\ell\times \ell})(k>\ell)$ is non-RS. Furthermore,  the code $\mathrm{GRL}_{k}(\boldsymbol{\alpha},\boldsymbol{v},\boldsymbol{A}_{\ell\times \ell})(k>\ell)$ is non-GRS.

This completes the proof of Theorem $\ref{nonRS}$.

$\hfill\Box$

\begin{remark}
$(1)$ If $k=\ell$ and $\boldsymbol{A}_{\ell\times \ell}$ is an $\ell\times \ell$ Vandermonde matrix $\begin{pmatrix}
1&\cdots&1\\
\beta_{1}&\cdots&\beta_{\ell}\\
\beta_{1}^{2}&\cdots&\beta_{\ell}^{2}\\
\vdots&\ddots&\vdots\\
\beta_{1}^{\ell-1}&\cdots&\beta_{\ell}^{\ell-1}
\end{pmatrix}$, then the corresponding code $\mathrm{GRL}_{k}(\boldsymbol{\alpha},\boldsymbol{v},\boldsymbol{A}_{\ell\times \ell})$ has the following generate matrix
\begin{equation}\label{k=ell}
\begin{pmatrix}
	1&\cdots&1&1&\cdots&1\\
	\alpha_{1}&\cdots&\alpha_{n}&\beta_{1}&\cdots&\beta_{\ell}\\
	\alpha_{1}^{2}&\cdots&\alpha_{n}^{2}&\beta_{1}^{2}&\cdots&\beta_{\ell}^{2}\\
	\vdots&\ddots&\vdots&\vdots&\ddots&\vdots\\
	\alpha_{1}^{\ell-1}&\cdots&\alpha_{n}^{\ell-1}&\beta_{1}^{\ell-1}&\cdots&\beta_{\ell}^{\ell-1}
\end{pmatrix},
\end{equation}
it's easy to know that the code $\mathrm{GRL}_{k}(\boldsymbol{\alpha},\boldsymbol{v},\boldsymbol{A}_{\ell\times \ell})$ generated by $(\ref{k=ell})$ is an $[n+\ell,\ell, n+1]$ RS code with the evaluation-point sequence $\boldsymbol{\alpha}=\left(\alpha_{1},\ldots,\alpha_{n},\beta_{1},\ldots,\beta_{\ell}\right)$.

$(2)$ If $k=\ell$ and $\boldsymbol{A}_{\ell\times \ell}$ is a $\ell\times \ell$ non-singular lower triangular matrix $\begin{pmatrix}
	a_{11}&0&\cdots&0\\
	*&a_{22}&\cdots&0\\
	\vdots&\vdots&\ddots&\vdots\\
	*&*&\cdots&a_{\ell\ell}
\end{pmatrix}$ with $a_{ii}\in\mathbb{F}_{q}^{*}$, then the corresponding code $\mathrm{GRL}_{k}(\boldsymbol{\alpha},\boldsymbol{v},\boldsymbol{A}_{\ell\times \ell})$ has the following generate matrix
\begin{equation}\label{k=ell2}
	\begin{pmatrix}
		1&\cdots&1&a_{11}&0&\cdots&0\\
		\alpha_{1}&\cdots&\alpha_{n}&*&a_{22}&\cdots&0\\ 
		\vdots&\ddots&\vdots&\vdots&\vdots&\ddots&\vdots\\
		\alpha_{1}^{\ell-1}&\cdots&\alpha_{n}^{\ell-1}&*&*&\cdots&a_{\ell\ell}
	\end{pmatrix},
\end{equation}
By the Magma program, for the code $\mathrm{GRL}_{k}(\boldsymbol{\alpha},\boldsymbol{v},\boldsymbol{A}_{\ell\times \ell})$ generated by $(\ref{k=ell2})$, we have 
$$\dim\left(\mathrm{GRL}_{k}^{2}(\boldsymbol{\alpha},\boldsymbol{v},\boldsymbol{A}_{\ell\times \ell})\right)>2k-1,$$
then by Proposition 1 of the reference \cite{nonGRS6}, the code $\mathrm{GRL}_{k}(\boldsymbol{\alpha},\boldsymbol{v},\boldsymbol{A}_{\ell\times \ell})$ generated by $(\ref{k=ell2})$ is a non-RS type.
\end{remark}
\section{Conclusions}\label{sec6}

In this paper, we study the LCD property of the code $\mathrm{GRL}_{k}(\boldsymbol{\alpha},\boldsymbol{1},\boldsymbol{A}_{\ell\times \ell})$ with both
Euclidean and Hermitian inner products. By taking some special vector $\boldsymbol{\alpha}=\left(\alpha_{1},\ldots,\alpha_{n}\right)$, we construct several classes of Euclidean LCD GRL codes, Hermitian LCD GRL codes, GRL codes with small-dimensional hull, respectively. And for several classes of Hermitian GRL codes, we firstly give an upper bound for the dimension of the hull. Subsequently, we apply the above results to construct several classes of EAQECCs. Finally, we prove that the code $\mathrm{GRL}_{k}(\boldsymbol{\alpha},\boldsymbol{1},\boldsymbol{A}_{\ell\times \ell})$ is non-GRS for $k>\ell$.

\section*{Acknowledgement}
This paper is supported by National Natural Science Foundation of China (Grant No. 12471494) and Natural Science Foundation of Sichuan Province (2024NSFSC2051). 

\section*{Appendix} 
\renewcommand{\theexample}{\Alph{section}.\arabic{example}}
\subsection*{Appendix A: Examples for Section \ref{sec3}}
\setcounter{section}{1}
\setcounter{example}{0}
In this section, we give some examples for Theorems \ref{EuclideanLCDRL1}, \ref{EuclideanLCDRL2}, \ref{EuclideanLCDRL31}, \ref{EuclideanLCDRL3} and \ref{EuclideanLCDRL4}, where Examples \ref{ELCDRL11}-\ref{ELCDRL12} are for Theorems \ref{EuclideanLCDRL1} (1)-(2), Examples \ref{ELCDRL21}-\ref{ELCDRL22} are for Theorem \ref{EuclideanLCDRL2} (1)-(2); Example \ref{ELCDRL311} is for Theorem \ref{EuclideanLCDRL31}; Example \ref{ELCDRL3245} (1)-(3) and Examples \ref{ELCDRL31}-\ref{ELCDRL33} are for Theorem \ref{EuclideanLCDRL3} (2),(4), (5), (1) and (3); Example \ref{ELCDRL41} is for Theorem \ref{EuclideanLCDRL4}, respectively.

\begin{example}\label{ELCDRL11}
	Let $q=3^4,k=5,\delta=2,\mathbb{F}_{q}^{*}=\langle\gamma\rangle,\boldsymbol{A}_{2\times 2}=\begin{pmatrix}
		\gamma&\gamma^2\\
		\gamma^3&\gamma^5
	\end{pmatrix}.$ By $\alpha_{i}=\gamma^{\frac{q-1}{k}i}=\gamma^{16i}$, we have $$\boldsymbol{\alpha}=\left(\gamma^{16+2},\gamma^{32+2},\gamma^{48+2},\gamma^{64+2},\gamma^{80+2}\right),$$ and then the corresponding GRL code $\mathcal{C}$ has the following generator matrix
	$$
	\begin{pmatrix}
		1&1&1&1&1&0&0\\
		\gamma^{18}&\gamma^{34}&\gamma^{50}&\gamma^{66}&\gamma^{82}&0&0\\
		\left(\gamma^{18}\right)^{2}&\left(\gamma^{34}\right)^{2}&\left(\gamma^{50}\right)^{2}&\left(\gamma^{66}\right)^{2}&\left(\gamma^{82}\right)^{2}&0&0\\
		\left(\gamma^{18}\right)^{3}&\left(\gamma^{34}\right)^{3}&\left(\gamma^{50}\right)^{3}&\left(\gamma^{66}\right)^{3}&\left(\gamma^{82}\right)^{3}&\gamma&\gamma^2\\
		\left(\gamma^{18}\right)^{4}&\left(\gamma^{34}\right)^{4}&\left(\gamma^{50}\right)^{4}&\left(\gamma^{66}\right)^{4}&\left(\gamma^{82}\right)^{4}&\gamma^3&\gamma^5
	\end{pmatrix}
	$$
	Based on the Magma program, $\mathcal{C}$ is Euclidean LCD MDS with the parameters $[7,5,3]_{3^4}$, which is consistent with Theorem \ref{EuclideanLCDRL1} (1).
\end{example}

\begin{example}\label{ELCDRL12}
	Let $q=5^2,k=8,\delta=1,\mathbb{F}_{q}^{*}=\langle\gamma\rangle,\boldsymbol{A}_{4\times 4}=\begin{pmatrix}
		1& \gamma& \gamma^2 & \gamma \\
		\gamma & \gamma^3 & \gamma^5 & \gamma^7 \\
		\gamma & \gamma^6 & \gamma^{10} & \gamma^{14} \\
		\gamma^{3} & \gamma^{9} & \gamma^{15} & \gamma^{21}
	\end{pmatrix}.$ By $\alpha_{i}=\gamma^{\frac{q-1}{k}i}=\gamma^{3i}$, we have $$\boldsymbol{\alpha}=\left(\gamma^{3+1},\gamma^{6+1},\gamma^{9+1},\gamma^{12+1},\gamma^{15+1},\gamma^{18+1},\gamma^{21+1},\gamma^{24+1}\right),$$ and then the corresponding GRL code $\mathcal{C}$ has the following generator matrix
	$$
	\begin{pmatrix}
		1&1&1&1&1&1&1&1&0&0&0&0\\
		\gamma^{4}&\gamma^{7}&\gamma^{10}&\gamma^{13}&\gamma^{16}&\gamma^{19}&\gamma^{22}&\gamma^{25}&0&0&0&0\\
		\left(\gamma^{4}\right)^2&\left(\gamma^{7}\right)^2&\left(\gamma^{10}\right)^2&\left(\gamma^{13}\right)^2&\left(\gamma^{16}\right)^2&\left(\gamma^{19}\right)^2&\left(\gamma^{22}\right)^2&\left(\gamma^{25}\right)^2&0&0&0&0\\
		\left(\gamma^{4}\right)^3&\left(\gamma^{7}\right)^3&\left(\gamma^{10}\right)^3&\left(\gamma^{13}\right)^3&\left(\gamma^{16}\right)^3&\left(\gamma^{19}\right)^3&\left(\gamma^{22}\right)^3&\left(\gamma^{25}\right)^3&0&0&0&0\\
		\left(\gamma^{4}\right)^4&\left(\gamma^{7}\right)^4&\left(\gamma^{10}\right)^4&\left(\gamma^{13}\right)^4&\left(\gamma^{16}\right)^4&\left(\gamma^{19}\right)^4&\left(\gamma^{22}\right)^4&\left(\gamma^{25}\right)^4&1& \gamma& \gamma^2 & \gamma\\
		\left(\gamma^{4}\right)^5&\left(\gamma^{7}\right)^5&\left(\gamma^{10}\right)^5&\left(\gamma^{13}\right)^5&\left(\gamma^{16}\right)^5&\left(\gamma^{19}\right)^5&\left(\gamma^{22}\right)^5&\left(\gamma^{25}\right)^5&\gamma& \gamma^3 & \gamma^5 & \gamma^7\\
		\left(\gamma^{4}\right)^6&\left(\gamma^{7}\right)^6&\left(\gamma^{10}\right)^2&\left(\gamma^{13}\right)^6&\left(\gamma^{16}\right)^6&\left(\gamma^{19}\right)^6&\left(\gamma^{22}\right)^6&\left(\gamma^{25}\right)^6&\gamma& \gamma^6 & \gamma^{10} & \gamma^{14}\\
		\left(\gamma^{4}\right)^7&\left(\gamma^{7}\right)^7&\left(\gamma^{10}\right)^7&\left(\gamma^{13}\right)^7&\left(\gamma^{16}\right)^7&\left(\gamma^{19}\right)^7&\left(\gamma^{22}\right)^7&\left(\gamma^{25}\right)^7&
		\gamma^{3} & \gamma^{9} & \gamma^{15} & \gamma^{21}\\
	\end{pmatrix}
	$$ 
	Based on the Magma program, $\mathcal{C}$ is Euclidean LCD NMDS with the parameters $[12,8,4]_{5^2}$, which is consistent with Theorem \ref{EuclideanLCDRL1} (2).
\end{example}

\begin{example}\label{ELCDRL21}
	Let $q=5^2,k=12,\delta=2,\mathbb{F}_{q}^{*}=\langle\gamma\rangle,\boldsymbol{A}_{2\times 2}=\begin{pmatrix}
	\gamma^2 & \gamma^4 & \gamma^6\\
	\gamma^2 & \gamma^5 & \gamma^7\\
	\gamma^5 & \gamma^8 & \gamma^9
	\end{pmatrix}.$ By $\alpha_{i}=\gamma^{\frac{q-1}{k}i}=\gamma^{2i}$, we have $$\boldsymbol{\alpha}=\left(0,\gamma^{2+2},\gamma^{4+2},\gamma^{6+2},\gamma^{8+2},\gamma^{10+2},\gamma^{12+2},\gamma^{14+2},\gamma^{16+2},\gamma^{18+2},\gamma^{20+2},\gamma^{22+2},\gamma^{24+2}\right),$$ and then the corresponding GRL code $\mathcal{C}$ has the following generator matrix
	$$
\begin{pmatrix}
	1 & 1 & 1 & 1 & 1 & 1 & 1 & 1 & 1 & 1 & 1 & 1 & 1 & 1 & 1 & 1 \\
	0 & \gamma^4 & \gamma^6 & \gamma^8 & \gamma^{10} & \gamma^{12} & \gamma^{14} & \gamma^{16} & \gamma^{18} & \gamma^{20} & \gamma^{22} & \gamma^{24} & \gamma^{26} & 0 & 0 & 0 \\
	0 & \gamma^8 & \gamma^{12} & \gamma^{16} & \gamma^{20} & \gamma^{24} & \gamma^{28} & \gamma^{32} & \gamma^{36} & \gamma^{40} & \gamma^{44} & \gamma^{48} & \gamma^{52} & 0 & 0 & 0 \\
	0 & \gamma^{12} & \gamma^{18} & \gamma^{24} & \gamma^{30} & \gamma^{36} & \gamma^{42} & \gamma^{48} & \gamma^{54} & \gamma^{60} & \gamma^{66} & \gamma^{72} & \gamma^{78} & 0 & 0 & 0 \\
	0 & \gamma^{16} & \gamma^{24} & \gamma^{32} & \gamma^{40} & \gamma^{48} & \gamma^{56} & \gamma^{64} & \gamma^{72} & \gamma^{80} & \gamma^{88} & \gamma^{96} & \gamma^{104} & 0 & 0 & 0 \\
	0 & \gamma^{20} & \gamma^{30} & \gamma^{40} & \gamma^{50} & \gamma^{60} & \gamma^{70} & \gamma^{80} & \gamma^{90} & \gamma^{100} & \gamma^{110} & \gamma^{120} & \gamma^{130} & 0 & 0 & 0 \\
	0 & \gamma^{24} & \gamma^{36} & \gamma^{48} & \gamma^{60} & \gamma^{72} & \gamma^{84} & \gamma^{96} & \gamma^{108} & \gamma^{120} & \gamma^{132} & \gamma^{144} & \gamma^{156} & 0 & 0 & 0 \\
	0 & \gamma^{28} & \gamma^{42} & \gamma^{56} & \gamma^{70} & \gamma^{84} & \gamma^{98} & \gamma^{112} & \gamma^{126} & \gamma^{140} & \gamma^{154} & \gamma^{168} & \gamma^{182} & 0 & 0 & 0 \\
	0 & \gamma^{32} & \gamma^{48} & \gamma^{64} & \gamma^{80} & \gamma^{96} & \gamma^{112} & \gamma^{128} & \gamma^{144} & \gamma^{160} & \gamma^{176} & \gamma^{192} & \gamma^{208} & 0 & 0 & 0 \\
	0 & \gamma^{36} & \gamma^{54} & \gamma^{72} & \gamma^{90} & \gamma^{108} & \gamma^{126} & \gamma^{144} & \gamma^{162} & \gamma^{180} & \gamma^{198} & \gamma^{216} & \gamma^{234} & \gamma^2 & \gamma^4 & \gamma^6 \\
	0 & \gamma^{40} & \gamma^{60} & \gamma^{80} & \gamma^{100} & \gamma^{120} & \gamma^{140} & \gamma^{160} & \gamma^{180} & \gamma^{200} & \gamma^{220} & \gamma^{240} & \gamma^{260} & \gamma^2 & \gamma^5 & \gamma^7 \\
	0 & \gamma^{44} & \gamma^{66} & \gamma^{88} & \gamma^{110} & \gamma^{132} & \gamma^{154} & \gamma^{176} & \gamma^{198} & \gamma^{220} & \gamma^{242} & \gamma^{264} & \gamma^{286} & \gamma^5 & \gamma^8 & \gamma^9
\end{pmatrix}
	$$
	Based on the Magma program, $\mathcal{C}$ is Euclidean LCD NMDS with the parameters $[16,12,4]_{5^2}$, which is consistent with Theorem \ref{EuclideanLCDRL2} (1).
\end{example}

\begin{example}\label{ELCDRL22}
	Let $q=5^2,k=8,\delta=1,\mathbb{F}_{q}^{*}=\langle\gamma\rangle,\boldsymbol{A}_{2\times 2}=\begin{pmatrix}
		1&\gamma\\
		\gamma^2&\gamma^4
	\end{pmatrix}.$ By $\alpha_{i}=\gamma^{\frac{q-1}{k}i}=\gamma^{3i}$, we have $$\boldsymbol{\alpha}=\left(0,\gamma^{3+1},\gamma^{6+1},\gamma^{9+1},\gamma^{12+1},\gamma^{15+1},\gamma^{18+1},\gamma^{21+1},\gamma^{24+1}\right),$$ and then the corresponding GRL code $\mathcal{C}$ has the following generator matrix
	$$
	\begin{pmatrix}
1 & 1 & 1 & 1 & 1 & 1 & 1 & 1 & 1 & 1 & 1 & 1 & 1 \\
0 & \gamma^4 & \gamma^7 & \gamma^{10} & \gamma^{13} & \gamma^{16} & \gamma^{19} & \gamma^{22} & \gamma^{25} & 0 & 0 & 0 & 0 \\
0 & \gamma^8 & \gamma^{14} & \gamma^{20} & \gamma^{26} & \gamma^{32} & \gamma^{38} & \gamma^{44} & \gamma^{50} & 0 & 0 & 0 & 0 \\
0 & \gamma^{12} & \gamma^{21} & \gamma^{30} & \gamma^{39} & \gamma^{48} & \gamma^{57} & \gamma^{66} & \gamma^{75} & 0 & 0 & 0 & 0 \\
0 & \gamma^{16} & \gamma^{28} & \gamma^{40} & \gamma^{52} & \gamma^{64} & \gamma^{76} & \gamma^{88} & \gamma^{100} & \gamma^1 & \gamma^3 & \gamma^5 & \gamma^7 \\
0 & \gamma^{20} & \gamma^{35} & \gamma^{50} & \gamma^{65} & \gamma^{80} & \gamma^{95} & \gamma^{110} & \gamma^{125} & \gamma^3 & \gamma^4 & \gamma^8 & \gamma^9 \\
0 & \gamma^{24} & \gamma^{42} & \gamma^{60} & \gamma^{78} & \gamma^{96} & \gamma^{114} & \gamma^{132} & \gamma^{150} & \gamma^{11} & \gamma^{12} & \gamma^{13} & \gamma^{15} \\
0 & \gamma^{28} & \gamma^{49} & \gamma^{70} & \gamma^{91} & \gamma^{112} & \gamma^{133} & \gamma^{154} & \gamma^{175} & \gamma^{16} & \gamma^{20} & \gamma^{21} & \gamma^{23}
\end{pmatrix}
	$$
	Based on the Magma program,  $\mathcal{C}$ is Euclidean LCD NMDS with the parameters $[13,8,5]_{5^2}$, which is consistent with Theorem \ref{EuclideanLCDRL2} (2).
\end{example}

\begin{example}\label{ELCDRL311}
	Let $q=31,k=5,s=1,t=9,\mathbb{F}_{q}^{*}=\langle\gamma\rangle,\boldsymbol{A}_{2\times 2}=\begin{pmatrix}
		\gamma&\gamma^2\\
		\gamma^3&\gamma^5
	\end{pmatrix},$ it's easy to see that $\frac{q-1}{k}=6\nmid t-s=8$ $$v_{2}(q-1)-v_{2}(k)-1=0\neq v_{2}(t-s)=3.$$
	Then we have 
	$$\boldsymbol{\alpha}=\left(\gamma^{6+1},\gamma^{12+1},\gamma^{18+1},\gamma^{24+1},\gamma^{30+1},\gamma^{6+9},\gamma^{12+9},\gamma^{18+9},\gamma^{24+9},\gamma^{30+9}\right),$$
	and then the corresponding GRL code $\mathcal{C}$ has the following generator matrix 
	$$
	\begin{pmatrix}
		1 & 1 & 1 & 1 & 1 & 1 & 1 & 1 & 1 & 1 &0&0 \\
		\gamma^{7} & \gamma^{13} & \gamma^{19} & \gamma^{25} & \gamma^{31} & \gamma^{15} & \gamma^{21} & \gamma^{27} & \gamma^{33} & \gamma^{39} &0 & 0 \\
		\left(\gamma^{7}\right)^{2} &\left( \gamma^{13}\right)^{2} & \left(\gamma^{19}\right)^{2} & \left(\gamma^{25}\right)^{2} & \left(\gamma^{31}\right)^{2} & \left(\gamma^{15}\right)^{2} & \left(\gamma^{21}\right)^{2} & \left(\gamma^{27}\right)^{2} & \left(\gamma^{33}\right)^{2} & \left(\gamma^{39}\right)^{2} & 0 & 0 \\ 
		\left(\gamma^{7}\right)^{3} &\left( \gamma^{13}\right)^{3} & \left(\gamma^{19}\right)^{3} & \left(\gamma^{25}\right)^{3} & \left(\gamma^{31}\right)^{3} & \left(\gamma^{15}\right)^{3} & \left(\gamma^{21}\right)^{3} & \left(\gamma^{27}\right)^{3} & \left(\gamma^{33}\right)^{3} & \left(\gamma^{39}\right)^{3} & \gamma & \gamma^2 \\
		\left(\gamma^{7}\right)^{4} &\left( \gamma^{13}\right)^{4} & \left(\gamma^{19}\right)^{4} & \left(\gamma^{25}\right)^{4} & \left(\gamma^{31}\right)^{4} & \left(\gamma^{15}\right)^{4} & \left(\gamma^{21}\right)^{4} & \left(\gamma^{27}\right)^{4} & \left(\gamma^{33}\right)^{4} & \left(\gamma^{39}\right)^{4} & \gamma^3 & \gamma^5 \\
	\end{pmatrix}.
	$$
	Based on the Magma program,  $\mathcal{C}$ is Euclidean LCD NMDS with the parameters $[12,5,7]_{31}$, which is consistent with Theorem \ref{EuclideanLCDRL31}.
\end{example}

\begin{example}\label{ELCDRL3245}
	Let $q=3^4,k=4,\mathbb{F}_{q}^{*}=\langle\gamma\rangle,\boldsymbol{A}_{2\times 2}=\begin{pmatrix}
		\gamma&\gamma^2\\
		\gamma^3&\gamma^5
	\end{pmatrix},$ it's easy to see that $$v_{2}(k)=2<v_{2}(q-1)=4$$
	and
	$$v_{5}(k)=0<v_{5}(q-1)=1,$$
	then we have the following three choices
	\begin{enumerate}[label=$(\arabic*)$]
		\item Let $\delta=2^{v_{2}(q-1)-v_{2}(k)-2}=1$. By $\alpha_{i}=\gamma^{\frac{q-1}{k}i}=\gamma^{20i}$, we have $$\boldsymbol{\alpha}=\left(\gamma^{20},\gamma^{40},\gamma^{60},\gamma^{80},\gamma^{21},\gamma^{41},\gamma^{61},\gamma^{81}\right),$$ and then the corresponding GRL code $\mathcal{C}$ has the following generator matrix
		$$
		\begin{pmatrix}
			1&1&1&1&1&1&1&1&0&0\\
			\gamma^{20}&\gamma^{40}&\gamma^{60}&\gamma^{80}&\gamma^{21}&\gamma^{41}&\gamma^{61}&\gamma^{81}&0&0\\
			\gamma^{40}&\gamma^{80}&\gamma^{120}&\gamma^{160}&\gamma^{42}&\gamma^{82}&\gamma^{122}&\gamma^{162}&\gamma&\gamma^{2}\\
			\gamma^{60}&\gamma^{120}&\gamma^{180}&\gamma^{240}&\gamma^{63}&\gamma^{123}&\gamma^{183}&\gamma^{243}&\gamma^{3}&\gamma^{5}\\
		\end{pmatrix}
		$$
		Based on the Magma program,  $\mathcal{C}$ is Euclidean LCD NMDS with the parameters $[10,4,6]_{3^4}$, which is consistent with Theorem \ref{EuclideanLCDRL3} (2).
		\item Let $\delta=2^{v_{2}(q-1)+1}=2^5$. By $\alpha_{i}=\gamma^{\frac{q-1}{k}i}=\gamma^{20i}$, we have $$\boldsymbol{\alpha}=\left(\gamma^{20},\gamma^{40},\gamma^{60},\gamma^{80},\gamma^{52},\gamma^{72},\gamma^{92},\gamma^{112}\right),$$ and then the corresponding GRL code $\mathcal{C}$ has the following generator matrix
		$$
		\begin{pmatrix}
			1&1&1&1&1&1&1&1&0&0\\
			\gamma^{20}&\gamma^{40}&\gamma^{60}&\gamma^{80}&\gamma^{52}&\gamma^{72}&\gamma^{92}&\gamma^{112}&0&0\\
			\gamma^{40}&\gamma^{80}&\gamma^{120}&\gamma^{160}&\gamma^{104}&\gamma^{144}&\gamma^{184}&\gamma^{224}&\gamma&\gamma^{2}\\
			\gamma^{60}&\gamma^{120}&\gamma^{180}&\gamma^{240}&\gamma^{156}&\gamma^{216}&\gamma^{276}&\gamma^{336}&\gamma^{3}&\gamma^{5}\\
		\end{pmatrix}
		$$
		Based on the Magma program,  $\mathcal{C}$ is Euclidean LCD MDS with the parameters $[10,4,7]_{3^4}$, which is consistent with Theorem \ref{EuclideanLCDRL3} (4).
		\item Let $\delta=5^{v_{2}(q-1)-v_{2}(k)-1}=1$. By $\alpha_{i}=\gamma^{\frac{q-1}{k}i}=\gamma^{20i}$, we have $$\boldsymbol{\alpha}=\left(\gamma^{20},\gamma^{40},\gamma^{60},\gamma^{80},\gamma^{21},\gamma^{41},\gamma^{61},\gamma^{81}\right),$$ and then the corresponding GRL code $\mathcal{C}$ has the following generator matrix
		$$
		\begin{pmatrix}
			1&1&1&1&1&1&1&1&0&0\\
			\gamma^{20}&\gamma^{40}&\gamma^{60}&\gamma^{80}&\gamma^{21}&\gamma^{41}&\gamma^{61}&\gamma^{81}&0&0\\
			\gamma^{40}&\gamma^{80}&\gamma^{120}&\gamma^{160}&\gamma^{42}&\gamma^{82}&\gamma^{122}&\gamma^{162}&\gamma&\gamma^{2}\\
			\gamma^{60}&\gamma^{120}&\gamma^{180}&\gamma^{240}&\gamma^{63}&\gamma^{123}&\gamma^{183}&\gamma^{243}&\gamma^{3}&\gamma^{5}\\
		\end{pmatrix}
		$$
		Based on the Magma program,  $\mathcal{C}$ is Euclidean LCD NMDS with the parameters $[10,4,6]_{3^4}$, which is consistent with Theorem \ref{EuclideanLCDRL3} (5).
		
	\end{enumerate}
	
\end{example}

\begin{example}\label{ELCDRL31}
	Let $q=5^2,k=8,\mathbb{F}_{q}^{*}=\langle\gamma\rangle,\boldsymbol{A}_{2\times 2}=\begin{pmatrix}
		1&\gamma\\
		\gamma^2&\gamma^4
	\end{pmatrix},$ it's easy to see that $$v_{2}(k)=v_{2}(q-1)=3.$$ 
	Then by taking $\delta=2^{1}$ and $\alpha_{i}=\gamma^{\frac{q-1}{k}i}=\gamma^{3i}$, we have $$\boldsymbol{\alpha}=\left(\gamma^{3},\gamma^{6},\gamma^{9},\gamma^{12},\gamma^{15},\gamma^{18},\gamma^{21},\gamma^{24},\gamma^{4},\gamma^{7},\gamma^{10},\gamma^{13},\gamma^{16},\gamma^{19},\gamma^{22},\gamma^{25}\right),$$ and then the corresponding GRL code $\mathcal{C}$ has the following generator matrix
	$$\begin{pmatrix}
		1&1&1&1&1&1&1&1&1&1&1&1&1&1&1&1&0&0\\
		\gamma^{3} & \gamma^{6} & \gamma^{9} & \gamma^{12} & \gamma^{15} & \gamma^{18} & \gamma^{21} & \gamma^{24} & \gamma^{4} & \gamma^{7} & \gamma^{10} & \gamma^{13} & \gamma^{16} & \gamma^{19} & \gamma^{22} & \gamma^{25}&0&0\\ 
		\gamma^{6} & \gamma^{12} & \gamma^{18} & \gamma^{24} & \gamma^{30} & \gamma^{36} & \gamma^{42} & \gamma^{48} & \gamma^{8} & \gamma^{14} & \gamma^{20} & \gamma^{26} & \gamma^{32} & \gamma^{38} & \gamma^{44} & \gamma^{50}&0&0\\ 
		\gamma^{9} & \gamma^{18} & \gamma^{27} & \gamma^{36} & \gamma^{45} & \gamma^{54} & \gamma^{63} & \gamma^{72} & \gamma^{12} & \gamma^{21} & \gamma^{30} & \gamma^{39} & \gamma^{48} & \gamma^{57} & \gamma^{66} & \gamma^{75}&0&0 \\ 
		\gamma^{12} & \gamma^{24} & \gamma^{36} & \gamma^{48} & \gamma^{60} & \gamma^{72} & \gamma^{84} & \gamma^{96} & \gamma^{16} & \gamma^{28} & \gamma^{40} & \gamma^{52} & \gamma^{64} & \gamma^{76} & \gamma^{88} & \gamma^{100}&0&0\\ 
		\gamma^{15} & \gamma^{30} & \gamma^{45} & \gamma^{60} & \gamma^{75} & \gamma^{90} & \gamma^{105} & \gamma^{120} & \gamma^{20} & \gamma^{35} & \gamma^{50} & \gamma^{65} & \gamma^{80} & \gamma^{95} & \gamma^{110} & \gamma^{125}&0&0 \\ 
		\gamma^{18} & \gamma^{36} & \gamma^{54} & \gamma^{72} & \gamma^{90} & \gamma^{108} & \gamma^{126} & \gamma^{144} & \gamma^{24} & \gamma^{42} & \gamma^{60} & \gamma^{78} & \gamma^{96} & \gamma^{114} & \gamma^{132} & \gamma^{150} &1&\gamma\\ 
		\gamma^{21} & \gamma^{42} & \gamma^{63} & \gamma^{84} & \gamma^{105} & \gamma^{126} & \gamma^{147} & \gamma^{168} & \gamma^{28} & \gamma^{49} & \gamma^{70} & \gamma^{91} & \gamma^{112} & \gamma^{133} & \gamma^{154} & \gamma^{175} &\gamma^2&\gamma^4\\
	\end{pmatrix}$$
	Based on the Magma program,  $\mathcal{C}$ is Euclidean LCD AMDS with the parameters $[18,8,10]_{5^2}$, which is consistent with Theorem \ref{EuclideanLCDRL3} (1).
\end{example}

\begin{example}\label{ELCDRL33}
	Let $q=7^2,k=6,\mathbb{F}_{q}^{*}=\langle\gamma\rangle,\boldsymbol{A}_{2\times 2}=\begin{pmatrix}
		\gamma&\gamma^2\\
		\gamma^3&\gamma^5
	\end{pmatrix},$ it's easy to see that $$v_{2}(q-1)-v_{2}(k)=2\neq 1$$
	and
	$$v_{3}(k)=v_{3}(q-1)=1.$$
	Then by taking $\delta=3^{1}$ and $\alpha_{i}=\gamma^{\frac{q-1}{k}i}=\gamma^{2i}$, we have 
	$$\boldsymbol{\alpha}=\left(\gamma^{2},\gamma^{4},\gamma^{6},\gamma^{8},\gamma^{10},\gamma^{12},\gamma^{5},\gamma^{7},\gamma^{9},\gamma^{11},\gamma^{13},\gamma^{15}\right),$$
	and then the corresponding GRL code $\mathcal{C}$ has the following generator matrix 
	$$
	\begin{pmatrix}
		1 & 1 & 1 & 1 & 1 & 1 & 1 & 1 & 1 & 1 & 1 & 1 & 0&0 \\
		\gamma^2 & \gamma^4 & \gamma^6 & \gamma^8 & \gamma^{10} & \gamma^{12} & \gamma^5 & \gamma^7 & \gamma^9 & \gamma^{11} & \gamma^{13} & \gamma^{15} & 0 & 0 \\
		\gamma^4 & \gamma^8 & \gamma^{12} & \gamma^{16} & \gamma^{20} & \gamma^{24} & \gamma^{10} & \gamma^{14} & \gamma^{18} & \gamma^{22} & \gamma^{26} & \gamma^{30} & 0 & 0 \\
		\gamma^6 & \gamma^{12} & \gamma^{18} & \gamma^{24} & \gamma^{30} & \gamma^{36} & \gamma^{15} & \gamma^{21} & \gamma^{27} & \gamma^{33} & \gamma^{39} & \gamma^{45} & 0 & 0 \\
		\gamma^8 & \gamma^{16} & \gamma^{24} & \gamma^{32} & \gamma^{40} & \gamma^{48} & \gamma^{20} & \gamma^{28} & \gamma^{36} & \gamma^{44} & \gamma^{52} & \gamma^{60} & \gamma & \gamma^2 \\
		\gamma^{10} & \gamma^{20} & \gamma^{30} & \gamma^{40} & \gamma^{50} & \gamma^{60} & \gamma^{25} & \gamma^{35} & \gamma^{45} & \gamma^{55} & \gamma^{65} & \gamma^{75} & \gamma^3 & \gamma^5 \\
	\end{pmatrix}.
	$$
	Based on the Magma program,  $\mathcal{C}$ is Euclidean LCD NMDS with the parameters $[14,6,8]_{7^2}$, which is consistent with Theorem \ref{EuclideanLCDRL3} (3).
\end{example}

\begin{example}\label{ELCDRL41}
	Let $q=7^2,k=12,\mathbb{F}_{q}^{*}=\langle\gamma\rangle,\boldsymbol{A}_{2\times 2}=\begin{pmatrix}
	\gamma&\gamma^2\\
	\gamma^3&\gamma^5
\end{pmatrix},$ it's easy to see that $$(3k,q)=1,$$
$$q-1=48\notin\left\{k,2k,3k\right\}=\left\{12,24,36\right\},$$
and $$v_{2}(q-1)=v_{2}(k)=4\geq 1.$$
	Then by taking $\alpha_{i}=\gamma^{\frac{q-1}{k}i}=\gamma^{4i}$, we have 
$$\boldsymbol{\alpha}=\left(\gamma^{4},\gamma^{8},\ldots,\gamma^{44},\gamma^{48},\gamma^{5},\gamma^{9},\ldots,\gamma^{45},\gamma^{49},\gamma^{6},\gamma^{10},\ldots,\gamma^{46},\gamma^{50}\right),$$
Based on the Magma program,  $\mathcal{C}$ is Euclidean LCD with the parameters $[50,12]_{7^2}$, which is consistent with Theorem \ref{EuclideanLCDRL4}.
\end{example}
\subsection*{Appendix B: Examples for Section \ref{sec4}} 
In this Appendix B, we give some examples for Theorems \ref{HermitianLCDRL1},\ref{HermitianLCDRL2},\ref{HermitianLCDRL3}, \ref{HermitianLCDRL31} and  \ref{HermitianLCDRL4}, where Example \ref{HLCDRL11} $(1)$-$(2)$ are for Theorem \ref{HermitianLCDRL1} $(1)$-$(2)$; Examples \ref{HLCDRL21}-\ref{HLCDRL22} are for Theorem \ref{HermitianLCDRL2} $(1)$-$(2)$; Examples \ref{HLCDRL32}-\ref{HLCDRL4} is for Theorems \ref{HermitianLCDRL3}, \ref{HermitianLCDRL31} and   \ref{HermitianLCDRL4}, respectively. 

\begin{example}\label{HLCDRL11}
	Let $q=3^2,k=8,\delta=1,\mathbb{F}_{q^2}^{*}=\langle\gamma\rangle.$ By $\alpha_{i}=\gamma^{\frac{q^2-1}{k}i}=\gamma^{10i}$, we have $$\boldsymbol{\alpha}=\left(\gamma^{11},\gamma^{21},\gamma^{31},\gamma^{41},\gamma^{51},\gamma^{61},\gamma^{71}\right),$$ and then we have the following two cases.
	
$(1)$ By taking $\boldsymbol{A}_{2\times 2}=\begin{pmatrix}
	\gamma&\gamma^2\\
	\gamma^3&\gamma^5
\end{pmatrix}.$ Then the corresponding GRL code $\mathcal{C}$ has the following generator matrix
$$
\begin{pmatrix}
1&1&1&1&1&1&1&1&0&0\\
\gamma^{11} & \gamma^{21} & \gamma^{31} & \gamma^{41} & \gamma^{51} & \gamma^{61} & \gamma^{71} & \gamma^{81} &0&0\\
(\gamma^{11})^2 & (\gamma^{21})^2 & (\gamma^{31})^2 & (\gamma^{41})^2 & (\gamma^{51})^2 & (\gamma^{61})^2 & (\gamma^{71})^2 & (\gamma^{81})^2 &0&0\\
(\gamma^{11})^3 & (\gamma^{21})^3 & (\gamma^{31})^3 & (\gamma^{41})^3 & (\gamma^{51})^3 & (\gamma^{61})^3 & (\gamma^{71})^3 & (\gamma^{81})^3 &0&0\\
(\gamma^{11})^4 & (\gamma^{21})^4 & (\gamma^{31})^4 & (\gamma^{41})^4 & (\gamma^{51})^4 & (\gamma^{61})^4 & (\gamma^{71})^4 & (\gamma^{81})^4 &0&0\\
(\gamma^{11})^5 & (\gamma^{21})^5 & (\gamma^{31})^5 & (\gamma^{41})^5 & (\gamma^{51})^5 & (\gamma^{61})^5 & (\gamma^{71})^5 & (\gamma^{81})^5 &0&0\\
(\gamma^{11})^6 & (\gamma^{21})^6 & (\gamma^{31})^6 & (\gamma^{41})^6 & (\gamma^{51})^6 & (\gamma^{61})^6 & (\gamma^{71})^6 & (\gamma^{81})^6  &\gamma&\gamma^2\\
(\gamma^{11})^7 & (\gamma^{21})^7 & (\gamma^{31})^7 & (\gamma^{41})^7 & (\gamma^{51})^7 & (\gamma^{61})^7 & (\gamma^{71})^7 & (\gamma^{81})^7 &
\gamma^3&\gamma^5
\end{pmatrix}
$$
Based on the Magma program,  $\mathcal{C}$ is Hermitian LCD MDS with the parameters $[10,8,3]_{3^4}$, which is consistent with Theorem \ref{HermitianLCDRL1} $(1)$. 

$(2)$ By taking $\boldsymbol{A}_{4\times 4}=\begin{pmatrix}
	\gamma & 1 & 1 & 1 \\
	1 & \gamma^2 & 1 & 1 \\
	1 & 1 & \gamma^{3} &1\\
	\gamma^4 & \gamma^{5} & \gamma^{6} & \gamma^{7}
\end{pmatrix}.$ Then the corresponding GRL code $\mathcal{C}$ has the following generator matrix
$$
\begin{pmatrix}
	1&1&1&1&1&1&1&1&0&0&0&0\\
	\gamma^{11} & \gamma^{21} & \gamma^{31} & \gamma^{41} & \gamma^{51} & \gamma^{61} & \gamma^{71} & \gamma^{81} &0&0&0&0\\
	(\gamma^{11})^2 & (\gamma^{21})^2 & (\gamma^{31})^2 & (\gamma^{41})^2 & (\gamma^{51})^2 & (\gamma^{61})^2 & (\gamma^{71})^2 & (\gamma^{81})^2 &0&0&0&0\\
	(\gamma^{11})^3 & (\gamma^{21})^3 & (\gamma^{31})^3 & (\gamma^{41})^3 & (\gamma^{51})^3 & (\gamma^{61})^3 & (\gamma^{71})^3 & (\gamma^{81})^3 &0&0&0&0\\
	(\gamma^{11})^4 & (\gamma^{21})^4 & (\gamma^{31})^4 & (\gamma^{41})^4 & (\gamma^{51})^4 & (\gamma^{61})^4 & (\gamma^{71})^4 & (\gamma^{81})^4 &\gamma & 1 & 1 & 1 \\
	(\gamma^{11})^5 & (\gamma^{21})^5 & (\gamma^{31})^5 & (\gamma^{41})^5 & (\gamma^{51})^5 & (\gamma^{61})^5 & (\gamma^{71})^5 & (\gamma^{81})^5 &1 & \gamma^2 & 1 & 1 \\
	(\gamma^{11})^6 & (\gamma^{21})^6 & (\gamma^{31})^6 & (\gamma^{41})^6 & (\gamma^{51})^6 & (\gamma^{61})^6 & (\gamma^{71})^6 & (\gamma^{81})^6  &1 &1 & \gamma^{3} & 1 \\
	(\gamma^{11})^7 & (\gamma^{21})^7 & (\gamma^{31})^7 & (\gamma^{41})^7 & (\gamma^{51})^7 & (\gamma^{61})^7 & (\gamma^{71})^7 & (\gamma^{81})^7 &
	\gamma^4 & \gamma^{5} & \gamma^{6} & \gamma^{7}
\end{pmatrix}
$$
Based on the Magma program,  $\mathcal{C}$ is Hermitian LCD NMDS with the parameters $[12,8,4]_{3^4}$, which is consistent with Theorem \ref{HermitianLCDRL1} $(2)$.
\end{example}

\begin{example}\label{HLCDRL21}
	Let $q=3^3,k=13,\delta=0,\mathbb{F}_{q^2}^{*}=\langle\gamma\rangle.$ By $\alpha_{i}=\gamma^{\frac{q^2-1}{k}i}=\gamma^{56i}$, we have $$\boldsymbol{\alpha}=\left(0,\gamma^{56},(\gamma^{56})^2,(\gamma^{56})^3,(\gamma^{56})^4,(\gamma^{56})^5,(\gamma^{56})^6,(\gamma^{56})^7,(\gamma^{56})^8,(\gamma^{56})^9,(\gamma^{56})^{10},(\gamma^{56})^{11},(\gamma^{56})^{12},(\gamma^{56})^{13}\right).$$
	By taking $\boldsymbol{A}_{2\times 2}=\begin{pmatrix}
		\gamma&\gamma^2\\
		\gamma^3&\gamma^5
	\end{pmatrix}.$ Then the corresponding GRL code $\mathcal{C}$ has the following generator matrix
	$$
	\begin{pmatrix}
1 & 1 & 1 & \dots & 1 & 0 & 0 \\
0 & \gamma^{56} & (\gamma^{56})^2 & \dots & (\gamma^{56})^{13} & 0 & 0 \\
0& (\gamma^{56})^2 & ((\gamma^{56})^2)^2 & \dots & ((\gamma^{56})^{13})^2 &0&0\\
0& (\gamma^{56})^3 & ((\gamma^{56})^2)^3 & \dots & ((\gamma^{56})^{13})^3 &0 &0 \\
\vdots & \vdots & \vdots & \ddots & \vdots & \vdots & \vdots \\
0& (\gamma^{56})^{10} & ((\gamma^{56})^2)^{10} & \dots & ((\gamma^{56})^{13})^{10} &0&0\\
0& (\gamma^{56})^{11} & ((\gamma^{56})^2)^{11} & \dots & ((\gamma^{56})^{13})^{11} & \gamma^1 & \gamma^2 \\
0& (\gamma^{56})^{12} & ((\gamma^{56})^2)^{12} & \dots & ((\gamma^{56})^{13})^{12} & \gamma^3 & \gamma^5 
\end{pmatrix}_{13\times 16}
	$$
	Based on the Magma program,  $\mathcal{C}$ is Hermitian LCD MDS with the parameters $[16,13,4]_{3^4}$, which is consistent with Theorem \ref{HermitianLCDRL2} $(1)$. 
\end{example}

\begin{example}\label{HLCDRL22}
Let $q=5^2,k=8,\delta=0,\mathbb{F}_{q^2}^{*}=\langle\gamma\rangle.$ By $\alpha_{i}=\gamma^{\frac{q^2-1}{k}i}=\gamma^{78i}$, we have $$\boldsymbol{\alpha}=\left(0,\gamma^{78},(\gamma^{78})^2,(\gamma^{78})^3,(\gamma^{78})^4,(\gamma^{78})^5,(\gamma^{78})^6,(\gamma^{78})^7,(\gamma^{78})^8\right).$$
By taking $\boldsymbol{A}_{2\times 2}=\begin{pmatrix}
	\gamma&\gamma^2\\
	\gamma^3&\gamma^5
\end{pmatrix}.$ Then the corresponding GRL code $\mathcal{C}$ has the following generator matrix
	$$
\small	\begin{pmatrix}
	1 & 1 & 1 & 1 & 1 & 1 & 1 & 1 & 1 & 0 & 0 & 0 & 0 \\
	0 & \gamma^{78} & (\gamma^{78})^2 & (\gamma^{78})^3 & (\gamma^{78})^4 & (\gamma^{78})^5 & (\gamma^{78})^6 & (\gamma^{78})^7 & (\gamma^{78})^8 & 0 & 0 & 0 & 0 \\
	0 & (\gamma^{78})^2 & ((\gamma^{78})^2)^2 & ((\gamma^{78})^3)^2 & ((\gamma^{78})^4)^2 & ((\gamma^{78})^5)^2 & ((\gamma^{78})^6)^2 & ((\gamma^{78})^7)^2 & ((\gamma^{78})^8)^2 & 0 & 0 & 0 & 0 \\
	0 & (\gamma^{78})^3 & ((\gamma^{78})^2)^3 & ((\gamma^{78})^3)^3 & ((\gamma^{78})^4)^3 & ((\gamma^{78})^5)^3 & ((\gamma^{78})^6)^3 & ((\gamma^{78})^7)^3 & ((\gamma^{78})^8)^3 & 0 & 0 & 0 & 0 \\
	0 & (\gamma^{78})^4 & ((\gamma^{78})^2)^4 & ((\gamma^{78})^3)^4 & ((\gamma^{78})^4)^4 & ((\gamma^{78})^5)^4 & ((\gamma^{78})^6)^4 & ((\gamma^{78})^7)^4 & ((\gamma^{78})^8)^4 & \gamma & 1 & 1 & 1 \\
	0 & (\gamma^{78})^5 & ((\gamma^{78})^2)^5 & ((\gamma^{78})^3)^5 & ((\gamma^{78})^4)^5 & ((\gamma^{78})^5)^5 & ((\gamma^{78})^6)^5 & ((\gamma^{78})^7)^5 & ((\gamma^{78})^8)^5 & 1 & \gamma^2 & 1 & 1 \\
	0 & (\gamma^{78})^6 & ((\gamma^{78})^2)^6 & ((\gamma^{78})^3)^6 & ((\gamma^{78})^4)^6 & ((\gamma^{78})^5)^6 & ((\gamma^{78})^6)^6 & ((\gamma^{78})^7)^6 & ((\gamma^{78})^8)^6 & 1 & 1 & \gamma^3 & 1 \\
	0 & (\gamma^{78})^7 & ((\gamma^{78})^2)^7 & ((\gamma^{78})^3)^7 & ((\gamma^{78})^4)^7 & ((\gamma^{78})^5)^7 & ((\gamma^{78})^6)^7 & ((\gamma^{78})^7)^7 & ((\gamma^{78})^8)^7 & \gamma^4 & \gamma^5 & \gamma^6 & \gamma^7
\end{pmatrix}.
$$
	Based on the Magma program,  $\mathcal{C}$ is Hermitian LCD NMDS with the parameters $[13,8,5]_{5^4}$, which is consistent with Theorem \ref{HermitianLCDRL2} $(2)$. 
\end{example}

\begin{example}\label{HLCDRL32}
	Let $q=13,k=6,\mathbb{F}_{q^2}^{*}=\langle\gamma\rangle,$ $$\boldsymbol{A}_{\ell\times \ell}\in L=\left\{\begin{pmatrix}
		\gamma&\gamma^2\\
		\gamma^3&\gamma^5
	\end{pmatrix},\begin{pmatrix}
	\gamma^1& \gamma^3& \gamma^5\\
	\gamma^3& \gamma^4& \gamma^8\\
	\gamma^{11}& \gamma^{12}& \gamma^{13}
	\end{pmatrix} \right\}.$$ It's easy to see that $\frac{q^2-1}{k}=28$ and $$0\notin T=\left\{v_{2}(q^2-1)-1-v_{2}(i+(k-i)q)|1\leq i\leq k-1\right\}=\left\{1\right\},$$
	then we take some $(s,t)$ satisfy $\frac{q^2-1}{k}\nmid s-t$ and $v_{2}\left(s-t\right)\notin T$, for example, $$(s,t)\in K=\left\{(4,1),(5,2),(10,1),(11,2),(28,1)\right\},$$ furthermore, we have $$\boldsymbol{\alpha}\in\left\{\left(\gamma^{28+s},\gamma^{56+s},\cdots,\gamma^{140+s},\gamma^{168+s},\gamma^{28+t},\gamma^{56+t},\cdots,\gamma^{140+t},\gamma^{168+t}\right)|(s,t)\in K\right\}.$$ 
	Then the corresponding GRL code $\mathcal{C}$ has the following generator matrix
	$$
	\begin{pmatrix}
	\begin{matrix}
		1&1&\cdots&1&1&1&1&\cdots&1&1\\
		\gamma^{28+s}&\gamma^{56+s}&\cdots&\gamma^{140+s}&\gamma^{168+s}&\gamma^{28+t}&\gamma^{56+t}&\cdots&\gamma^{140+t}&\gamma^{168+t}\\
		\left(\gamma^{28+s}\right)^{2}&\left(\gamma^{56+s}\right)^{2}&\cdots&\left(\gamma^{140+s}\right)^{2}&\left(\gamma^{168+s}\right)^{2}&\left(\gamma^{28+t}\right)^{2}&\left(\gamma^{56+t}\right)^{2}&\cdots&\left(\gamma^{140+t}\right)^{2}&\left(\gamma^{168+t}\right)^{2}\\ 
		
		\left(\gamma^{28+s}\right)^{3}&\left(\gamma^{56+s}\right)^{3}&\cdots&\left(\gamma^{140+s}\right)^{3}&\left(\gamma^{168+s}\right)^{3}&\left(\gamma^{28+t}\right)^{3}&\left(\gamma^{56+t}\right)^{3}&\cdots&\left(\gamma^{140+t}\right)^{3}&\left(\gamma^{168+t}\right)^{3}\\
		
		\left(\gamma^{28+s}\right)^{4}&\left(\gamma^{56+s}\right)^{4}&\cdots&\left(\gamma^{140+s}\right)^{4}&\left(\gamma^{168+s}\right)^{4}&\left(\gamma^{28+t}\right)^{4}&\left(\gamma^{56+t}\right)^{4}&\cdots&\left(\gamma^{140+t}\right)^{4}&\left(\gamma^{168+t}\right)^{4}\\
		
		\left(\gamma^{28+s}\right)^{5}&\left(\gamma^{56+s}\right)^{5}&\cdots&\left(\gamma^{140+s}\right)^{5}&\left(\gamma^{168+s}\right)^{5}&\left(\gamma^{28+t}\right)^{5}&\left(\gamma^{56+t}\right)^{5}&\cdots&\left(\gamma^{140+t}\right)^{5}&\left(\gamma^{168+t}\right)^{5}
	\end{matrix}&\begin{matrix}
	\boldsymbol{0}_{(6-\ell)\times\ell}\\
	\\
	\\
	\boldsymbol{A}_{\ell\times \ell}
\end{matrix}\end{pmatrix}
	$$
with $\boldsymbol{A}_{\ell\times \ell}\in L.$ Based on the Magma program, when $\boldsymbol{A}_{\ell\times \ell}=\begin{pmatrix}
	\gamma&\gamma^2\\
	\gamma^3&\gamma^5
\end{pmatrix}$,  $\mathcal{C}$ is Hermitian LCD NMDS with the parameters $[14,6,8]_{13^2}$, which is consistent with Theorem \ref{HermitianLCDRL3};
when $\boldsymbol{A}_{\ell\times \ell}=\begin{pmatrix}
	\gamma^1& \gamma^3& \gamma^5\\
	\gamma^3& \gamma^4& \gamma^8\\
	\gamma^{11}& \gamma^{12}& \gamma^{13}
\end{pmatrix}$,  $\mathcal{C}$ is Hermitian LCD NMDS with the parameters $[15,6,9]_{13^2},$ expect for the case $(s,t)=(5,2)$, which is Hermitian LCD with the parameters $[15,6,8]_{13^2}$, which is consistent with Theorem \ref{HermitianLCDRL3}.
\end{example}

\begin{example}\label{HLCDRL3}
	Let $q=11,k=5,\mathbb{F}_{q^2}^{*}=\langle\gamma\rangle,\boldsymbol{A}_{2\times 2}=\begin{pmatrix}
		\gamma&\gamma^2\\
		\gamma^3&\gamma^5
	\end{pmatrix}.$ It's easy to see that $\frac{q^2-1}{k}=24$ and $$T=\left\{v_{2}(q^2-1)-1-v_{2}(i+(k-i)q)|1\leq i\leq k-1\right\}=\left\{2\right\},$$
	then we take some $\delta$ satisfy $\frac{q^2-1}{k}\nmid \delta$ and $v_{2}\left(\delta\right)\notin T$, for example, $\delta=2,3,6,8$, furthermore, we have $$\boldsymbol{\alpha}\in\left\{\left(\gamma^{24},\gamma^{48},\gamma^{72},\gamma^{96},\gamma^{120},\gamma^{24+\delta},\gamma^{48+\delta},\gamma^{72+\delta},\gamma^{96+\delta},\gamma^{120+\delta}\right)|\delta=2,3,6,8\right\}.$$ 
Then the corresponding GRL code $\mathcal{C}$ has the following generator matrix
$$
\begin{pmatrix}
1&1&1&1&1&1&1&1&1&1&0&0\\
\gamma^{24}&\gamma^{48}&\gamma^{72}&\gamma^{96}&\gamma^{120}&\gamma^{24+\delta}&\gamma^{48+\delta}&\gamma^{72+\delta}&\gamma^{96+\delta}&\gamma^{120+\delta}&0&0\\
\left(\gamma^{24}\right)^{2}&\left(\gamma^{48}\right)^{2}&\left(\gamma^{72}\right)^{2}&\left(\gamma^{96}\right)^{2}&\left(\gamma^{120}\right)^{2}&\left(\gamma^{24+\delta}\right)^{2}&\left(\gamma^{48+\delta}\right)^{2}&\left(\gamma^{72+\delta}\right)^{2}&\left(\gamma^{96+\delta}\right)^{2}&\left(\gamma^{120+\delta}\right)^{2}&0&0\\ 
\left(\gamma^{24}\right)^{3}&\left(\gamma^{48}\right)^{3}&\left(\gamma^{72}\right)^{3}&\left(\gamma^{96}\right)^{3}&\left(\gamma^{120}\right)^{3}&\left(\gamma^{24+\delta}\right)^{3}&\left(\gamma^{48+\delta}\right)^{3}&\left(\gamma^{72+\delta}\right)^{3}&\left(\gamma^{96+\delta}\right)^{3}&\left(\gamma^{120+\delta}\right)^{3}&\gamma&\gamma^2\\
\left(\gamma^{24}\right)^{4}&\left(\gamma^{48}\right)^{4}&\left(\gamma^{72}\right)^{4}&\left(\gamma^{96}\right)^{4}&\left(\gamma^{120}\right)^{4}&\left(\gamma^{24+\delta}\right)^{4}&\left(\gamma^{48+\delta}\right)^{4}&\left(\gamma^{72+\delta}\right)^{4}&\left(\gamma^{96+\delta}\right)^{4}&\left(\gamma^{120+\delta}\right)^{4}&\gamma^3&\gamma^5
\end{pmatrix}
$$
Based on the Magma program,  $\mathcal{C}$ is Hermitian LCD NMDS with the parameters $[12,5,7]_{11^2}$, which is consistent with Theorem \ref{HermitianLCDRL31}. 
\end{example}

\begin{example}\label{HLCDRL4}
	Let $q=11,k=5,\mathbb{F}_{q^2}^{*}=\langle\gamma\rangle,\boldsymbol{A}_{2\times 2}=\begin{pmatrix}
		\gamma&\gamma^2\\
		\gamma^3&\gamma^5
	\end{pmatrix}.$ It's easy to know that
	$$\left\{\frac{q^2-1}{\left(q^2-1,i+(k-i)q\right)}\middle|1\leq i\leq k-1\right\}=\left\{8,24\right\},$$
	then we take some $\delta$ satisfy $\left(\delta k,q\right)=1$ and $\frac{q^2-1}{\left(q^2-1,i+(k-i)q\right)}\nmid \delta+1$ for any $1\leq i\leq k-1$, for example, $\delta=1,2,3,4,5,6$, furthermore, we have $$\boldsymbol{\alpha}\in\left\{\left(\gamma^{24},\gamma^{48},\gamma^{72},\gamma^{96},\gamma^{120},\gamma^{25},\gamma^{49},\gamma^{73},\gamma^{97},\gamma^{121},\ldots,\gamma^{24+\delta},\gamma^{48+\delta},\gamma^{72+\delta},\gamma^{96+\delta},\gamma^{120+\delta}\right)|\delta=1,2,3,4,5,6\right\}.$$ 
	Then the corresponding GRL code $\mathcal{C}$ has the following generator matrix
$$
\footnotesize\begin{pmatrix}
	1&\cdots&1&1&\cdots&1&\cdots&1&\cdots&1&0&0\\
	\gamma^{24}&\cdots&\gamma^{120}&\gamma^{24+1}&\cdots&\gamma^{120+1}&\cdots&\gamma^{24+\delta}&\cdots&\gamma^{120+\delta}&0&0\\
	\left(\gamma^{24}\right)^{2}&\cdots&\left(\gamma^{120}\right)^{2}&\left(\gamma^{24+1}\right)^{2}&\cdots&\left(\gamma^{120+1}\right)^{2}&\cdots&\left(\gamma^{24+\delta}\right)^{2}&\cdots&\left(\gamma^{120+\delta}\right)^{2}&0&0\\ 
	\left(\gamma^{24}\right)^{3}&\cdots&\left(\gamma^{120}\right)^{3}&\left(\gamma^{24+1}\right)^{3}&\cdots&\left(\gamma^{120+1}\right)^{3}&\cdots&\left(\gamma^{24+\delta}\right)^{3}&\cdots&\left(\gamma^{120+\delta}\right)^{3}&\gamma&\gamma^2\\
	\left(\gamma^{24}\right)^{4}&\cdots&\left(\gamma^{120}\right)^{4}&\left(\gamma^{24+1}\right)^{4}&\cdots&\left(\gamma^{120+1}\right)^{4}&\cdots&\left(\gamma^{24+\delta}\right)^{4}&\cdots&\left(\gamma^{120+\delta}\right)^{4}&\gamma^3&\gamma^5
\end{pmatrix}
$$
	Based on the Magma program,  $\mathcal{C}$ is Hermitian LCD NMDS with the parameters $$\left\{[12,5,7]_{11^2},[17,5,12]_{11^2},[22,5,17]_{11^2},[27,5,22]_{11^2},[32,5,27]_{11^2},[37,5,32]_{11^2}\right\},$$
which is consistent with Theorem \ref{HermitianLCDRL4}.
\end{example}
\end{document}